\documentclass{llncs}
\usepackage{llncsdoc}

\usepackage{amsmath} 
\usepackage[T1]{fontenc}
\usepackage{algorithm}
\usepackage[noend]{algpseudocode}
\usepackage{amssymb}        
\usepackage{latexsym}
\usepackage{array}
\usepackage{graphicx}

\newtheorem{fact}[theorem]{Fact}

\usepackage{amsopn}

\author{Pratik Ghosal \and Adam Kunysz \and Katarzyna Paluch}

\date{}

\institute{University of Wroc\l aw, Wroc\l aw, Poland}
\title{The Dynamics of Rank-Maximal and Popular Matchings\protect\footnote{{
{This work was partially funded by Polish National Science Center
    grant UMO-2018/29/B/ST6/02633}}}}

\begin{document}
\sloppy
\maketitle
\begin{abstract}
Given a bipartite graph, where the two sets of vertices are applicants and posts and ranks on the edges represent preferences of applicants over posts, a {\em rank-maximal} matching is one in which the maximum number of applicants is matched to their rank one posts and subject to this condition, the maximum number of applicants is matched to their rank two posts, and so on. 
We study the dynamic version of the problem in which a new applicant or post may be added to the graph and we would like to maintain a rank-maximal matching. We show that after the arrival of one vertex, we are always able to update the existing rank-maximal matching in $\mathcal{O}(\min(c'n ,n^2) + m)$ time, where  $n$ denotes the number of applicants, $m$ the number of edges and $c'$ the maximum rank of an edge in an optimal solution.  Additionally, we update the matching using a minimal  number of changes (replacements). All cases of a deletion of a vertex/edge and an addition of an edge   can be reduced to the problem of handling the addition of a vertex.
As a by-product, we also get  an  analogous $\mathcal{O}(m)$ result for  the dynamic version of the (one-sided) popular matching problem.  

Our results are based on the novel use of the properties of the Edmonds-Gallai decomposition. The presented ideas may find applications
in other (dynamic) matching problems.

\end{abstract}

\section{Introduction}

We consider the dynamic version of the rank-maximal matching problem. In the rank-maximal matching problem, we are given a bipartite graph
$G = (\mathcal{A} \cup \mathcal{P}, \mathcal{E})$, where $\mathcal{A}$ is a set of applicants, $\mathcal{P}$  a set of posts and edges have ranks. An edge $(a,p)$ has rank $i$ if
the post $p$ is one of the applicant $a$'s $i$th choices. A matching  of the graph $G$ is said to be {\em rank-maximal} if it matches the maximum number of applicants  to their rank one posts and subject to this condition, it matches the maximum number of applicants  to their rank two posts, and so on. A rank-maximal matching can be computed in $\mathcal{O}(\min(c \sqrt{n},n) m)$ time, where $n$ denotes the number of applicants, $m$ the number of edges and $c$ the maximum rank of an edge in an optimal solution \cite{IrvingKMMP06}. The algorithm from \cite{IrvingKMMP06} consists in successive computations of a maximum matching $M_i$ of a so-called reduced graph $G'_i$. The reduced graph $G'_i$
contains an appropriately trimmed set of edges of rank at most $i$ and the computation of the maximum matching $M_i$ is always conducted by
extending the previously found maximum matching $M_{i-1}$ of the reduced graph $G'_{i-1}$.
Rank-maximal matchings have applications in assigning papers to referees \cite{GargKKMM10}, projects to students etc.

In the dynamic variant of the problem a new vertex may be added to the graph and we would like to maintain a rank-maximal matching.
When the new vertex $v$ is added to the graph $G$ we assume that the graph $G$ itself does not change. In particular, if a new post $p$ arrives,
the applicants of $G$ cannot change their preferences over the posts that are already included in $G$. Let us call the graph $G$ extended by $v$ and the edges incident to $v$ as the graph $H$. In order to have a rank-maximal matching of $H$, we would like to be able to transform a rank-maximal matching $M$ of $G$ into a rank-maximal matching $N$ of $H$, making only the smallest needed number of changes. In some cases
a rank-maximal matching of $G$ is also  rank-maximal in $H$. We design an algorithm that  updates $M$ by an application of {\em only one} alternating path  $P$, i.e., $M \oplus P= (M \setminus P) \cup (P \setminus M)$  is a rank-maximal matching of $H$. To be able to compute $P$ efficiently, we need access to the reduced graphs $G'_1, G'_2, \ldots, G'_c$ of $G$
(the notion is defined in \cite{IrvingKMMP06} and also recalled in Section \ref{not}) and their Edmonds-Gallai decompositions. The reduced graphs and their decompositions can be stored in $\mathcal{O}(n^2+m)$ space. We show that we can compute a required alternating path $P$ as well as update the reduced graphs $H'_1, H'_2, \ldots, H'_{c'}$ of $H$ and their Edmonds-Gallai decompositions in $\mathcal{O}(\min(c'n ,n^2) + m)$ time ($c'$ is defined analogously). 
The time bound can be considered optimal under the circumstances, as improving it would imply a better  running time for the rank-maximal matching problem.

The result may seem rather surprising in the sense that we are able to compress $r$ phases, each of which requires the computation of a matching and the update of the Edmonds-Gallai decomposition, into one phase with the running time of 
$\mathcal{O}(\min(c'n ,n^2) + m)$. For comparison, let us note that it is much easier to update the matching gradually  - separately in each of the graphs $H_i$ that consists of edges of rank at most $i$. 
In such an approach, however, it is required to compute and apply $r$ alternating paths. Each such computation and update of the reduced graph
$H'_i$ can be carried out in $\mathcal{O}(n+m)$ time and thus the overall running time is $\mathcal{O}(r(n+m))$.  This is how the problem is dealt with in a recent paper by  Nimbhorkar and Rameshwar \cite{nimbhorkar2018dynamic}. We, instead, abstain from updating the matching until the last phase when we have collected all the necessary data in the form of a carefully built alternating subgraph $T$. 
 This subgraph is rooted at the new vertex $v$ and contains all possible
alternating paths, whose application results in a rank-maximal matching of $H$.
To be able to  efficiently build this subgraph and hence update a rank-maximal matching, we identify new properties of the Edmonds-Gallai decomposition, that are of independent interest and are potentially applicable to other (dynamic) matching problems. We want to observe that even a simple checking whether the matching  needs to be updated at all is not easy to carry out faster than in $\mathcal{O}(cm)$ time. In the paper we show how it can be done in $\mathcal{O}(m)$ time. For the case when the matching requires to be 
altered, one of the main new ideas that allows us to do so in a good time bound is that of recognizing the alternating paths that will finally belong to the alternating subgraph $T$ and ignoring those that will not. 

Observe that our algorithm is significantly faster than the one in \cite{nimbhorkar2018dynamic} - the improvement is always of the order of $\Omega(m/n)$ and may be even   $\Omega(m)$ for the case when $r$ is of the order of $\Omega(m)$. (In standard settings $r$ is   $\mathcal{O}(n)$, however,  each edge in the graph may be assigned a different rank (because we give priorities to certain applicants) and then the number of  distinct ranks may be $\Omega(m)$.)   Additionally, to update the matching we use a  minimal  number of  changes (replacements). To achieve this, from all alternating paths $P'$ such that $M \oplus P'$ is a rank-maximal matching of the new graph, we select  the shortest one. 


We present an algorithm for the version, in which a new applicant is added to the graph. This solution applies to the situation 
when a new post  arrives - note that ranks are assigned to edges and, from the point of view of the algorithm there is no real difference between applicants and posts (there is one, however, in their interpretation). We show, that all cases of: a deletion of a vertex from the graph, an addition 
or a deletion of a new edge or even a change of the rank of a given edge can be reduced to the problem of handling the addition of a vertex.

The \emph{popular matching} problem in the one-sided version is defined as follows. The input is the same as in the rank-maximal matching problem - we are given a bipartite graph $G$, in which the vertices of one side of the graph express their preferences over the vertices of the other side. The goal is to find a {\em popular} matching
in $G$, if it exists. A matching $M$ is said to be popular if there exists no other matching $M'$ such that $M'$  is {\em more popular} than $M$. A matching $M'$ is more popular than $M$ if the number of applicants preferring $M'$ to $M$ is greater than the number of applicants preferring $M$ to $M'$ and an applicant $a$ prefers $M'$ to $M$ if (i) he is matched in $M'$ and unmatched in $M$ or (ii) he prefers the post $M'(a)$ to $M(a)$. Not every instance of the problem admits  a popular matching. Nevertheless, Abraham et al. \cite{AbrahamIKM07} gave an $\mathcal{O}(\sqrt{n}m)$ time algorithm that computes a popular matching, if it exists.
The algorithm is in a certain sense similar to the one computing a rank-maximal matching. It consists of two phases that are the same as in the algorithm for rank-maximal  matchings, but the edges participating in the second phase are defined on the basis of so-called first and second posts. To put it differently, every popular matching of $G$ is an applicant-complete rank-maximal matching of a subgraph of $G$ that only contains  edges connecting each applicant to its first and second posts. To obtain a solution for the dynamic version of the popular 
matching problem, we can thus directly use the algorithm for the dynamic version of the rank-maximal matching problem. Since the number of phases is two,  we are able to update a popular matching in $\mathcal{O}(m)$ time after the arrival/deletion of a new vertex/edge.  Nimbhorkar and Rameshwar \cite{nimbhorkar2018dynamic} are also able to update a popular matching in $\mathcal{O}(m)$ time, however, they have no control over the number of applied changes.

The algorithm for updating a rank-maximal matching can be also used for updating a bounded unpopularity matching in the same time bound of $\mathcal{O}(\min(c'n ,n^2) + m)$ \cite{HuangKMN11}.


{\bf Previous work} A rank-maximal matching can be found via a relatively straightforward reduction to a maximum weight matching. The running time 
of the resulting algorithm is $\mathcal{O}(r^2 \sqrt{n}m\log n)$, where $r$ denotes the maximal rank of an edge, if we use the Gabow-Tarjan \cite{GT89} algorithm, or $\mathcal{O}(rn(m + n \log n))$ for the Fredman-Tarjan algorithm \cite{FT87}. The first algorithm for rank-maximal matchings was given by Irving in \cite{Irvgreedy} for the version without ties and with the 
running time of $\mathcal{O}(d^2 n^3)$, where $d$ denotes the maximum degree of an applicant in the graph (thus $d \leq r$). The already mentioned \cite{IrvingKMMP06} gives
a combinatorial algorithm that runs in $\mathcal{O}(\min(n, c \sqrt n)m)$ time. The capacitated and weighted versions were considered, respectively,
in \cite{Paluch13} and \cite{KavithaS06}. A switching graph characterisation of the set of all rank-maximal matchings is described in \cite{GhoshalNN14}.  Independently of our work, in a recent paper \cite{nimbhorkar2018dynamic} Nimbhorkar and Rameshwar also study the dynamic version of the rank-maximal matching problem and develop an $\mathcal{O}(r(n+m))$ algorithm for updating a rank-maximal matching after the addition or deletion of a vertex or edge.

{\bf Related Work}
Matchings under preferences in the dynamic setting have been studied under different notions of optimality. In \cite{DBLP:conf/latin/McCutchen08} McCutchen introduced the notion of an unpopularity factor and showed that it is NP-hard to compute a least unpopular matching in one-sided instances. Bhattacharya et. al. \cite{BhattacharyaHHK15} gave an algorithm to maintain matchings with an unpopularity factor of $(\Delta + k)$ by making an amortized number of $\mathcal{O}(\Delta+\Delta^2/k)$ changes per round, for any $k > 0$ where $\Delta$ denotes the maximum degree of any applicant in any round. Note that this is the number of changes made to the matching and not the update time, which is much higher and requires a series of computations of a maximum weight matching.

In \cite{DBLP:journals/siamdm/AbrahamK10} Abraham and Kavitha describe the notion of a so-called \emph{voting path}. A voting path is a sequence of matchings which starts from an arbitrary matching, and ends at a popular matching and each matching in the sequence is more popular than the previous one. The authors showed that in the one-sided setting with ties there always exists a voting path of length at most two. They also show how to compute such paths in linear time,  given a popular matching in the graph, which allows them to maintain a popular matching under a sequence of deletions and additions of vertices to the graph, however, in $(\mathcal{O}(\sqrt{n}m)$ time per each update.

Pareto optimality is another well-known criterion. In \cite{AbrahamCMM05} authors gave an $\mathcal{O}(\sqrt{n}m)$ time algorithm for computing Pareto optimal matchings. In \cite{FleischerW} Fleischer and Wang studied Pareto optimal matchings in the dynamic setting. The authors gave a linear time algorithm to maintain a maximum size Pareto matching under a sequence of deletions and additions of vertices.

{\bf Organization} Section \ref{not} recalls the definitions and the rank-maximal matching algorithm. Section \ref{simple} contains a description of the simplified variant of the problem, in which we only want to check if the update is necessary, i.e., if the rank-maximal matching of $G$ is also a rank-maximal matching of the new graph $H$. In Section \ref{aug} we describe the ideas behind the alternating subgraph $T$
that contains all paths, whose application to the current matching yields an updated rank-maximal matching. In Section \ref{alg} we present the 
algorithm for updating the rank-maximal matching and give the proof of its correctness. Section $\ref{example}$ contains two examples illustrating the algorithm presented in Section \ref{alg}. In Section \ref{updates} we present an algorithm that efficiently updates reduced graphs  after applying Algorithm \ref{main}. Finally in Section \ref{popular} we present an algorithm for the dynamic popular matching problem.
\section{Preliminaries} \label{not}

Let $G = (\mathcal{A} \cup \mathcal{P}, \mathcal{E})$ be a bipartite graph and let $M$ be a maximum matching of $G$. 
We say that a path is {\em $M$-alternating} if its edges belong alternately to $M$ and $\mathcal{E}\setminus M$. We say that a vertex $v$ is
{\em free} or {\em unmatched} in $M$ if no edge of $M$ is incident to $v$. An $M$-alternating path is said to be {\em $M$-augmenting} (or {\em augmenting} if the matching is clear from the context) if it starts and ends at an unmatched vertex.

By $V(G)$ and $\mathcal{E}(G)$ we denote, respectively, the set of vertices of $G$ and the set of edges.

Given a maximum matching $M$, we can partition the vertex set of $G$ into three disjoint sets $E$, $O$ and $U$. Vertices in $E$, $O$ and $U$ are called \emph{even}, \emph{odd} and \emph{unreachable} respectively and are defined as follows. A vertex $v \in V(G)$ is \emph{even} (resp. \emph{odd}) if there is an even (resp. odd) length alternating path in $G$ with respect to $M$ from an unmatched vertex to $v$. A vertex $v \in V(G)$ is unreachable if there is no alternating path in $G$ with respect to $M$ from an unmatched vertex to $v$. For vertex sets $A$ and $B$, we call an edge connecting a vertex in $A$ with a vertex in $B$ an $AB$ edge. 

The following lemma is well known in matching theory.

\begin{lemma} \textbf{Edmonds-Gallai decomposition (EG-decomposition)} \cite{sch}, \cite{IrvingKMMP06}
\label{EG}
Let $M$ be a maximum matching in $G$ and let $E$, $O$ and $U$ be defined as above.
\begin{enumerate}
	\item \label{EG1} The sets $E$, $O$, $U$ are pairwise disjoint.
	\item \label{EG2} Let $N$ be any maximum matching in $G$. \begin{enumerate}
		\item \label{EG2A} $N$ defines the same sets $E$, $O$ and $U$.
		\item \label{EG2B} $N$ contains only $UU$ and $OE$ edges.
		\item \label{EG2C} Every vertex in $O$ and every vertex in $U$ is matched by $N$.
		\item \label{EG2D} $|N| = |O| + |U|/2$.
	\end{enumerate}
	\item \label{EG3}There is no $EU$ and no $EE$ edge in $G$.
\end{enumerate} 

\end{lemma}

Throughout the paper we consider many graphs at once, thus to avoid confusion, for a given graph $G$ we denote the sets of even, odd and unreachable vertices as $E(G)$, $O(G)$ and $U(G)$ respectively.

\subsection{Rank-Maximal Matchings}
Next we review an  algorithm by Irving et al. \cite{IrvingKMMP06} for computing  a rank-maximal matching. Let $G = (\mathcal{A} \cup \mathcal{P}, \mathcal{E})$
be an instance of the rank-maximal matching problem. Every edge $e=(a,p)$ has a rank reflecting its position in the preference list of applicant $a$.  $\mathcal{E}$ is the union of disjoint sets $\mathcal{E}_i$ , i.e.,  $\mathcal{E} = \mathcal{E}_1 \cup \mathcal{E}_2 \cup \mathcal{E}_3 ... \cup \mathcal{E}_r$, where $\mathcal{E}_i$ denotes the set of edges of rank $i$.

\begin{definition}\cite{IrvingKMMP06}
The signature of a matching $M$ is defined as an $r$-tuple $\rho(M) = (x_1,..., x_r)$ where, for each $1 \leq i \leq r$, $x_i$ is the number of
applicants who are matched to their $i$-th rank post in $M$.
\end{definition}

Let $M$ and $M'$ be two matchings of $G$, with the signatures $\rho(M) = (x_1,..., x_r)$ and
$\rho(M') = (y_1,..., y_r)$. We say $M \succ M'$ if there exists $1 \leq k \leq r$ such that $x_k > y_k$ and $x_i = y_i$ for each $i < k$ with $i \in \mathbb{N}$. 

\begin{definition}
A matching $M$ of a graph $G$ is called  rank-maximal  if and only if $M$ has the best signature under the ordering $\succ$ defined above.
\end{definition}

We give a brief description of the algorithm of Irving et al. \cite{IrvingKMMP06} for computing a rank-maximal matching, whose pseudocode (Algorithm \ref{alg1}) is given below.  Let us denote $G_i = (\mathcal{A} \cup \mathcal{P}, \mathcal{E}_1 \cup \mathcal{E}_2 \cup ...\cup \mathcal{E}_i)$ as a subgraph of $G$ that only contains edges of rank smaller or equal to $i$. $G'_i$ is called the reduced graph of $G_i$ for $1 \leq i \leq r$. The algorithm runs in phases. The algorithm starts with $G'_1 = G_1$ and a maximum matching $M_1$ of $G_1$.  In the first phase, the set of vertices is partitioned into $E(G'_{1})$, $O(G'_{1})$ and $U(G'_{1})$. The edges between $O(G'_{1})$ and $O(G'_{1}) \cup U(G'_{1})$ are deleted. Since any vertex in $ O(G'_{1}) \cup U(G'_{1})$ has to be matched in $G_1$ in every rank-maximal matching, the edges of rank greater than $1$ incident to such vertices are deleted from the graph $G$.  Next we add the edges of rank $2$ and call the resulting graph $G'_2$. The graph $G'_2$ may contain some $M_1$-augmenting paths. We determine the maximum matching $M_2$ in $G'_2$  by augmenting $M_1$.
In the $i$-th phase,  the vertices are partitioned into three disjoint sets $E(G'_i)$, $O(G'_i)$ and $U(G'_i)$. We delete every edge between $O(G'_{i})$ and $O(G'_{i}) \cup U(G'_{i})$. Also, we delete every edge of rank greater than $i$ incident to  vertices in $O(G'_{i})\cup U(G'_{i})$. Next we add the edges of rank $(i+1)$ and call the resulting graph $G'_{i+1}$. We determine the maximum matching $M_{i+1}$ in $G'_{i+1}$ by  augmenting $M_i$.

The pseudocode of Irving et al.'s algorithm \cite{IrvingKMMP06} is denoted as Algorithm \ref{alg1}.

\begin{theorem} \cite{IrvingKMMP06}
Algorithm \ref{alg1} computes a rank-maximal matching in $\mathcal{O}(\min\{c\sqrt{n}, n \}m)$ time, where $c \leq r$ denotes a maximal rank in the optimal solution.
\end{theorem}

\begin{algorithm}[h]
\caption{for computing a rank-maximal matching}
\label{alg1}
\begin{algorithmic}[1]
\State $G'_1 \gets G_1$
\State Let $M_1$ be any maximum matching of $G'_1$.
\For {$i = 1, 2, \ldots, r$}
	\State Determine a partition of the vertices of $G'_i$ into the sets $E(G'_i)$, $O(G'_i)$ and $U(G'_i)$.
	\State Delete all edges in $\mathcal{E}_j$ (for $j > i$) which are incident on nodes in $O(G'_i) \cup U(G'_i)$.\State Delete all $O(G'_i)O(G'_i)$ and $O(G'_i)U(G'_i)$ edges from $G'_i$.
	\State Add the edges in $\mathcal{E}_{i+1}$ and call the resulting graph $G'_{i+1}$. 
	\State Determine a maximum matching $M_{i+1}$ in $G'_{i+1}$ by augmenting $M_i$.
	\EndFor
	\Return $M_{r}$
\end{algorithmic}
\end{algorithm}

The following invariants of Algorithm \ref{alg1} are proven in \cite{IrvingKMMP06}.

\begin{enumerate}  
	\item For every $1 \leq i \leq r$, every rank-maximal matching in $G_i$ is contained in $G'_i$.
	\item The matching $M_i$ is rank-maximal in $G_i$, and is a maximum matching of $G'_i$.
	\item If a rank-maximal matching in $G$ has signature $(s_1, s_2, \ldots, s_i, \ldots, s_r)$ then $M_i$ has signature $(s_1, s_2, \ldots, s_i)$. 
	\item The graphs $G'_i$ ($1 \leq i \leq r$) constructed during the execution of Algorithm \ref{alg1} are independent of the rank-maximal matching computed by the algorithm.
\end{enumerate}

We say that a vertex $v$ is {\em alive} in $G'_i$ iff $v \in \bigcap_{j=1}^{i-1} E(G'_j)$.
$Alive(i)$ denotes the set of vertices that are alive in $G'_i$.

\begin{fact}\label{alive}
Each edge of $G'_i \setminus G'_{i-1}$ ($1 \leq i \leq r$) has both endpoints in $Alive(i)$.
\end{fact} 
This follows from line $5$  of Algorithm \ref{alg1}.

Algorithm \ref{alg1} can be modified so that it terminates in $c$ iterations, where $c$ is the maximum rank of an edge in an optimal solution. We simply stop when there are no more edges to add. It is shown in  \cite{Paluch13} that the last iteration, in which edges are added is iteration $c$. Observe also that by Fact \ref{alive} in every iteration, in which $G'_i$ contains edges of rank $i$, matching $M_i$ is augmented  and thus contains at least one edge of rank $i$.

\subsection{The Dynamic Rank-Maximal Matching Problem}

In the dynamic variant of the rank-maximal matching problem, we are given a graph $G$ and we wish to maintain a rank-maximal matching of this graph under a sequence of the following kinds of operations:

\begin{enumerate}
	\item Add a vertex $v$ along with the edges incident to it to $G$.
	\item Delete a vertex $v$ along with the edges incident to it from $G$.
	\item Add an edge $e$ to $G$.
	\item Delete an edge $e$ from $G$.
\end{enumerate}

In order to perform the above operations efficiently, we additionally maintain all structures which are normally computed by Algorithm \ref{alg1}, i.e. the reduced graphs $G'_i$ along with their $EG$-decompositions and the matchings $M_i$. Let us denote the modified graph obtained from $G$ after performing one of the operations $1-4$ by $H = (\mathcal{A}' \cup \mathcal{P}', \mathcal{F}_1 \cup \mathcal{F}_2 \cup \ldots \cup \mathcal{F}_{r})$, where $\mathcal{F}_i$ consists of the edges of rank $i$ in $H$. Similarly for each $1 \leq i \leq r$ denote: $H_i = (\mathcal{A}' \cup \mathcal{P}', \mathcal{F}_1 \cup \mathcal{F}_2 \cup \ldots \cup \mathcal{F}_i)$ and $H'_i = (\mathcal{A}' \cup \mathcal{P}', \mathcal{F}'_1 \cup \mathcal{F}'_2 \cup \ldots \cup \mathcal{F}'_i)$.  Our goal is to compute a rank-maximal matching of $H$ along with all the reduced graphs $H'_i$ and the matchings $N_i$ (where $N_i$ is a rank-maximal matching of $H_i$). Note that we do not execute Algorithm \ref{alg1} on $H$ but update the existing graphs $G'_i$ in order to obtain $H'_i$ and the matchings $M_i$ in order to obtain $N_i$. Also, before  finding the reduced graphs $H'_i$, we first compute graphs $\tilde{H}_i=(\mathcal{A}' \cup \mathcal{P}', \mathcal{\tilde{F}}'_1 \cup \mathcal{\tilde{F}}'_2 \cup \ldots \cup \mathcal{\tilde{F}}'_i)$ that may differ slightly from graphs $H'_i$. Each graph $\tilde{H}_i$
has the property that every edge of any rank-maximal matching of $H_i$ is contained in $\tilde{H}_i$ and $\tilde{H}_1=H_1$. 

It turns out that we do not actually need to construct separate algorithms for each of the operations $1-4$. As we show in Section \ref{remainingUpdates} only the operation $1$ is truly needed. We prove that the remaining updates can be simulated with a constant number of executions of the operation $1$.  In the remainder of the paper, we focus on the implementation of the operation $1$.

It is easy to observe that an algorithm that maintains a rank-maximal matching after adding an applicant is symmetrical to the case where we add a post. Hence, without loss of generality, in the remainder of the paper, we can assume that a new applicant $a_0$ arrives and we want to maintain a rank-maximal matching after adding that applicant.

\section{Algorithm for Checking if Update is Necessary} \label{simple}
Before we describe our algorithm for maintaining a rank-maximal matching under a sequence of operations of type $1$, we first introduce and solve a simplified variant of the problem. The main goal of this section is to build some intuition.

Our first assumption is that a newly arriving applicant $a_0$ has only one edge incident on it. We also slightly change our goal. Instead of computing a rank-maximal matching of $H$ we only wish to determine if a rank-maximal matching $M$ of $G$ remains rank-maximal in $H$. Our goal is to solve this problem in $\mathcal{O}(m)$ time. The following is the main theorem of this section:

\begin{theorem}
Assume that we are given reduced graphs $G'_1, G'_2, \ldots, G'_{r+1}$ of $G$, their EG-decompositions and matchings $M_1$, $M_2, \ldots, M_{r}$. Then there is an $\mathcal{O}(m)$ time algorithm that determines if $M_r$ is a rank-maximal matching of $H$.
\end{theorem}

Let us first describe the main idea behind Algorithm \ref{checkingalgorithm}. From Invariant 2 of Algorithm \ref{alg1}, it directly follows that if $M_r$ is not a rank-maximal matching of $H$, then there exists $j$ such that $M_j$ is not a rank-maximal matching of $H_j$ and for each $i < j$ matching $M_i$ is rank-maximal in $H_i$. From the same invariant, it follows that $H'_j$ contains a larger maximum matching than $M_j$. Our goal is to iterate over $i = 1,2,\ldots,r$ and for each $i$ to determine the structure of the reduced graph $H'_i$. Based on the structure of $H'_i$, we simply check whether $H'_i$ contains a larger matching than $M_i$. If in some iteration $j$, we determine that $M_j$ is not a maximum matching of $H'_j$ then obviously $M_r$ is not a rank-maximal matching of $H$. Otherwise we claim that $M_r$ remains rank-maximal in $H$.

Note that if we follow the above approach, in the worst case we have to check whether $M_i$ is a maximum matching of $H'_i$ for each $1 \leq i \leq r$. Since we are interested in an $\mathcal{O}(m)$ time algorithm we cannot afford to compute each $H'_i$ from scratch as in Algorithm \ref{alg1}. We claim that since $G$ and $H$ differ only by one edge for each $i$, we can construct the graph $H'_i$ based on $G'_i$ and $H'_{i-1}$. Additionally, we can also check whether $M_i$ is a maximum matching of $H'_i$ based on the $EG$-decomposition of $G'_i$. 

In the following auxiliary lemma, we examine how the maximum matching $M$ in a bipartite graph $G$ and the EG-decomposition of $G$ change when we add one edge to the graph. 

We say that a vertex $v$ has {\em type} even, odd or unreachable in $G$ if $v\in E(G), \ v\in O(G)$ or $v \in U(G)$, respectively. 
Similarly, we say that $v$ has the same type in $G$ and $J$ if ($v \in X(G)$ iff $v \in X(J)$), where $X \in \{E, O, U \}$.

\begin{lemma} \label{lemGE}
Let $G = (\mathcal{A} \cup \mathcal{P}, \mathcal{E})$ be a bipartite graph, $M$ a maximum matching of $G$ and  $a\in \mathcal{A}$ and $p \in \mathcal{P}$ two vertices of $G$ such that $(a,p) \notin \mathcal{E}$ and $a$ has type $E$  in the EG-decomposition of $G$ ($a \in E(G)$). Then the graph $J= (\mathcal{A} \cup \mathcal{P}, \mathcal{E} \cup (a,p))$ has the following properties:
\begin{enumerate}
	\item If $p \in E(G)$, then the edge $(a,p)$ belongs to every maximum matching of $J$. A maximum matching of $J$ is of size $|M|+1$.
	\item If $p \in O(G)$, then the edge $(a,p)$ belongs to some maximum matching of $J$ but not to every one and $M$ remains a maximum matching of $J$. Additionally,  the EG-decomposition of the graph $J$ is the same as that of $G$. 
	\item If $p \in U(G)$, then the edge $(a,p)$ belongs to some maximum matching of $J$ but not to every one and $M$ remains a maximum matching of $J$. Additionally,  the EG-decomposition of the graph $J$ is different from that of $G$ in the following way. A vertex $v \in U(G)$ belongs to $E(J)$ (respectively, $O(J)$) if there exists an even-length (corr., odd-length) alternating path starting from the vertex $a$ that contains the edge $(a,p)$ and ends at $v$. Apart from this every vertex has the same type in the EG-decompositions of $G$ and $J$.
\end{enumerate}

\end{lemma} 
\begin{proof} 

We first prove $(1)$. Since we have $a, p \in E(G)$, from the properties of Edmonds-Gallai decomposition there exist alternating paths $P_1$ and $P_2$ in $G$ with respect to $M$ from free vertices $v_1, v_2$ ending in respectively $a$ and $p$. $v_1$ and $v_2$ must be distinct otherwise  the alternating paths  $P_1$ and $P_2$ and the edge $(a,p)$ creates an odd cycle. It can be easily shown that we can combine $P_1$ and $P_2$ to obtain an augmenting path from $v_1$ to $v_2$ containing $(a, p)$. This implies that any maximum matching of $J$ is of size $|M| + 1$ and $(a, p)$ belongs to every maximum matching of $J$.

Let us now prove $(2)$. We first show that $M$ is a maximum matching of $J$. Let us assume by contradiction that it is not true. Then in $J$ there exists an augmenting path $P_1$ from a free vertex $x_1$ to another free vertex $x_2$ with respect to $M$. The path $P_1$ contains $(a, p)$ as otherwise $M$ would not be a maximum matching of $G$. Let us consider a subpath of this path which does not contain $(a, p)$ but contains $p$. Such a subpath is of course of even length and is contained in $G$. This implies that $p \in E(G)$, which leads to a contradiction. Thus $M$ is a maximum matching of $J$.

Let us now consider an alternating path $P$ of even length from a free vertex that contains $(a, p)$ and ends with the matched edge incident to $p$. Note that $M \oplus P$ is a maximum matching of $J$ containing $(a, p)$.

We now prove that $EG$-decompositions of $G$ and $J$ are identical. Let $v \in E(G)$. From the properties of $EG$-decomposition in $G$ there exists an alternating path of even length from a vertex $x_0$ to $v$ with respect to $M$. Such a path is also contained in $J$ thus we have $v \in E(J)$. This implies that $E(G) \subseteq E(J)$. We can similarly show that $O(G) \subseteq O(J)$. To prove that $EG$-decompositions are identical it suffices to show that $U(G)  \subseteq U(J)$. Let us now assume by contradiction that there exists $v \in U(G)$ such that $v \notin U(J)$. 
Let $P$ be an $M$-alternating path from a free vertex to $v$ in the graph $J$. $P$ must contain the edge $(a,p)$, otherwise the path is also present in $G$. Let $P_1$ denote the even length subpath of $P$ from the free vertex to $a$ and $P_2$ the subpath between $p$ and $v$. Clearly, both $P_1$ and $P_2$ appear in $G$.

Let $P_2 = \{p = v_1, v_2, \ldots, v_n = v\}$ be the alternating path where $p \in O(G)$ and $v \in U(G)$. Let $i$ be the smallest index such that $v_i \in U(G)$. Then $v_{i-1} \in E(G) \cup O(G)$. If $v_{i-1} \in E(G)$, then by Lemma \ref{EG} point \ref{EG3} $v_i \in O(G)$. If $v_{i-1} \in O(G)$ then by the construction of $P_2$, $(v_{i-1}, v_i) \in M$. Then by Lemma \ref{EG} point \ref{EG2B}  $v_i \in E(G)$. In both cases, we arrive at a   contradiction.  Therefore $v \in O(G) \cup E(G)$.

It remains to show $(3)$. The majority of the proof is analogous to $(2)$. We can similarly prove that $M$ is a maximum matching of $J$ and that $(a, p)$ belongs to some maximum matching of $J$ but not to all of them. Analogously we show that $E(G) \subseteq E(J)$ and $O(G) \subseteq O(J)$. 
Suppose $v \in U(G)$ but $v \notin U(J)$. Without loss of generality, assume that $v \in E(J)$. We prove that there is an even length alternating path from the vertex $a$ to $v$. Since $v \in E(J)$ there is an even length alternating path $P$ from a free vertex to $v$ in $J$. Clearly $P$ contains the edge $(a,p)$ and let us define $P_1$ as the even length subpath of $P$ from the free vertex to $a$. Thus $P \setminus P_1$ is an even length alternating path from $a$ to $v$. Conversely, let there be an even length alternating path $P_2$ between $a$ and $v$. Since $a \in E(J)$ there is an even length alternating path $P_1$ between a free vertex and $a$. Consequently, $P_1 \cup P_2$ contains an even length alternating path from a free vertex to $v$ in $J$. Therefore $v \in E(J)$. 
The proof is analogous if $v \in U(G) \cap O(J)$. \qed
\end{proof}


Based on the above lemma, we can determine if a maximum matching of $H$ is larger than a maximum matching of $G$. If maximum matchings of $G$ and $H$ are of the same size, then we can obtain the $EG$-decomposition of $H$ from the $EG$-decomposition of $G$. If $p \in O(G)$ both $EG$-decompositions are identical. If we have $p \in U(G)$ we can easily update the $EG$-decomposition of $G$ to the $EG$-decomposition of $H$ by a simple breadth-first search along alternating paths from  the edge $(a,p)$.

Below we describe Algorithm \ref{checkingalgorithm} in more details. In particular, we show how to apply Lemma \ref{lemGE} in order to efficiently obtain $H'_{i}$ from graphs $H'_{i-1}$ and $G'_i$.

Let us assume that the newly added edge $(a_0, p_0)$ is of rank $k$. From the pseudocode of Algorithm \ref{alg1}, we can see that for each $i$ such that $1 \leq i < k$ we have $G'_i = H'_i$, and that $M_i$ is a rank-maximal matching of $H_i$. How do graphs $G'_k$ and $H'_k$ differ? One can easily see that either $G'_k + (a_0, p_0) = H'_k$ or $G'_k = H'_k$ holds. The latter case happens when the edge $(a_0, p_0)$ is removed from $\mathcal{F}_k$. It can only happen if in some iteration $j < k$ we have $p_0 \notin E(G'_j)$. 

From now on we assume that $p_0 \in \bigcap_{i=1}^{k-1} E(G'_i)$. One can check that when we enter the loop $for$ in  line $4$ of Algorithm \ref{alg1} we have $G'_k + (a_0, p_0) = H'_k$. We can use Lemma \ref{lemGE} to obtain the information about the EG-decomposition of $H'_k$ from the decomposition of $G'_k$. From the statement of Lemma \ref{lemGE} it follows that there are three cases depending on the type of $p_0$ in $G'_k$.

Case $(1)$ - $p_0 \in E(G'_k)$. We can simply halt the algorithm and claim that $M$ is not a rank-maximal matching of $H$.

Case $(2)$ - $p_0 \in U(G'_k)$. From Lemma \ref{lemGE} we can see that some vertices may belong to $U(G'_k) \cap E(H'_k)$ or $U(G'_k) \cap O(H'_k)$. If a vertex $v \in U(G'_k) \cap E(H'_k)$  (resp. $v \in U(G'_k) \cap O(H'_k)$) then we say that $v$ changes its type from $U$ to $E$ (resp. $O$) in phase $k$. What implications does this fact have on the execution of Algorithm \ref{alg1} on $H$? Note that in lines $5$ and $6$ of
Algorithm \ref{alg1}, we remove some edges incident to vertices of types $O$ and $U$. If $v$ changes its type from $U$ to $E$ in phase $k$ then  the edges incident to $v$ that are deleted in phase $k$ during the execution of Algorithm \ref{alg1} in $G$, are not deleted in $H$. 
Such edges become \emph{activated} and in the pseudocode we denote the set of these edges as $AE_u$. Additionally, vertices which change type from either $U$ to $E$ are called \emph{activated vertices}. The set of such vertices is denoted as $AV$.

Case $(3)$ - $p_0 \in O(G'_k)$. We already know from Lemma \ref{lemGE} that the presence of $(a_0, p_0)$ in $H'_k$ does not affect its EG-decomposition. It turns out however that if for some $k' > k$ we have $p_0 \in U(G'_{k'})$ but $p_0 \in O(G'_{k'-1})$ then the presence of $(a_0, p_0)$ in $H$ might affect the EG-decomposition of $H'_{k'}$, but will not have any impact on the decompositions of graphs $H'_{l}$ for $k < l < k'$. Such edges also become \emph{activated} and added to $AE_o$.

The main idea behind the remaining part of the algorithm is to maintain the set $AE$ of \emph{activated edges} so that in any phase $k' > k$ a reduced graph $H'_{k'}$ is obtained from $G'_{k'}$ by adding the activated edges to this graph. The EG-decomposition of $H'_{k'}$ is then computed with the aid of decompositions of $G'_{k'}$ and $H'_{k'-1}$.  It is important to note that in phase $k$ graphs $G'_k$ and $H'_k$ differ by exactly one edge which allows us to apply Lemma \ref{lemGE}, whereas in  phase $k'$ ($k' > k$) $H'_{k'}$ may potentially contain multiple activated edges. We simply apply Lemma \ref{lemGE} to each activated edge in order to determine if $M_{k'}$ is a maximum matching of $H'_{k'}$.

The correctness of the algorithm follows from the above discussion and Lemma \ref{lemGE}. It is also included in Theorem \ref{cor}.

In the pseudocode of the algorithm  the subgraph $C$ contains a new vertex $a_0$ and  vertices that are at this point unreachable in $G$ (contained $U(G'_i)$) but belonging to $E(H'_i) \cup O(H'_i)$. Thus each activated vertex belongs to $C$ and $C$ contains
(the "upper") part of the alternating subgraph $T$ mentioned in the introduction. $R$ represents the rest of the graph - vertices that have the same type in $G'_i$ and $H'_i$.

\begin{algorithm} 
\caption{for checking if $M_r$ is a rank-maximal matching of $H$ } \label{checkingalgorithm}
\begin{algorithmic}[1]
\State $C \leftarrow \{a_0\}$, $AV \leftarrow \{a_0\}$, $AE \leftarrow \emptyset$
\State $i \leftarrow 1$
\While {$i \leq r$}
	\State $R \leftarrow G'_i \setminus C$
	 \ForAll{ $a \in AV$  }
	 \State $AE \leftarrow AE \cup \{(a,p)\in \mathcal{F}_i: a\in AV \wedge p \in Alive(i)\}$
	 \EndFor
	\If { there exists $(a, p) \in AE$ such that  $p \in E(G'_i)$}
	  \Return $M_r$ is not a rank-maximal matching of $H$
	\Else
		\State $H'_i \leftarrow C \cup R \cup AE$
		\State $AE_u \leftarrow \{(a,p) \in AE: \ p \in U(G'_i)\}$
    \ForAll{$S$ : $S$ is an even-length $M_i$-alt. path in $H'_i$ between $a_0$ and $a \in U(G'_i)$}
		\State $V(C) \leftarrow V(C) \cup V(S), \ \   \mathcal{E}(C) \leftarrow \mathcal{E}(C) \cup \{(a,p) \in G'_i: a, p \in S\}$
		\State $AV \leftarrow AV \cup \{a\}$
		\EndFor
		\State $AE_o \leftarrow \{(a,p) \in AE \cup G'_i: a \in C \wedge \ p \in O(G'_i)\}$
		\State $AE \leftarrow AE \cup AE_o \setminus AE_u$
		\EndIf
	\State $i \leftarrow i+1$
\EndWhile
 \Return $M_r$ is a rank-maximal matching of $H$
\end{algorithmic}
\end{algorithm}

The following example (Figure \ref{section3}) illustrates the implementation of Algorithm \ref{checkingalgorithm}.  In this example we check if a rank-maximal matching of $G$ is also a rank-maximal matching of $H$  or not. Here $p_0$ is an alive vertex in iteration $3$, hence we can add $(a_0, p_0)$ of rank $3$ to $G$. $a_1$ is an activated vertex in the third iteration. Both $(a_1, p_1)$ and $(a_1, p_2)$ are the activated edges of rank $4$ incident to $a_1$. Note that $p_1 \in U(G'_4)$ and $p_2 \in O(G'_4)$. Hence we add the edge $(a_1, p_1)$ to $AE_u$ and $(a_1, p_2)$ to $AE_o$ after iteration $4$. $p_2$ becomes an unreachable vertex after iteration $6$, we move $p_2$ to $AE_u$ after this iteration.

There is no iteration $1 \leq i \leq 6$, such that we have an edge $(a,p)$ incident to the activated vertex $a$ and $p \in E(G'_i)$. Therefore, we can conclude that a rank-maximal matching of $G$ is indeed a rank-maximal matching of $H$. 
\begin{figure}
\centering
   \includegraphics[width=.8\textwidth]{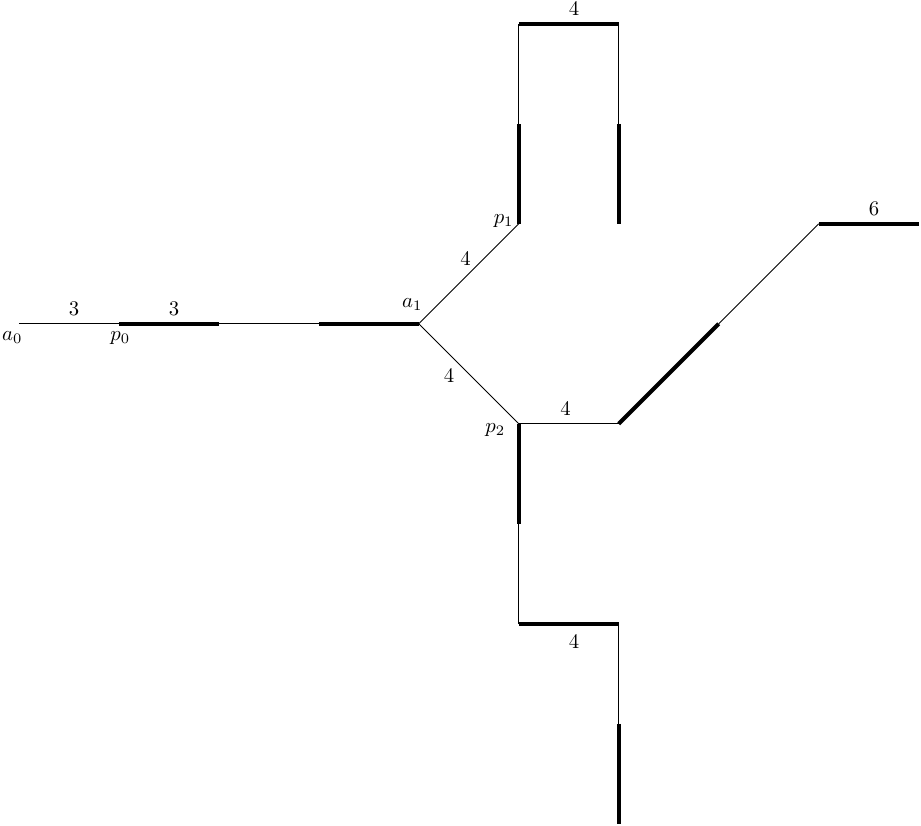}
	\caption{An example illustrating the implementation of Algorithm \ref{checkingalgorithm}}
	\label{section3}
\end{figure}
\section{Overview of the Algorithm} \label{aug}

In this section, we present some of the ideas behind Algorithm \ref{main} for updating a rank-maximal matching. The algorithm is essentially an extension of Algorithm \ref{checkingalgorithm}. The main difference is that at some point  Algorithm \ref{checkingalgorithm} in line $8$ may encounter an edge $(v,w)$ such that $v$ belongs to $U(G'_i) \cap E(H'_i)$   and $w$ belongs to $E(G'_i)$ and then
output \emph{"$M_r$ is not a rank-maximal matching of $H$"}. 

If we encounter such a situation in phase $i$, we have to compute matchings $N_i$, $N_{i+1}$, \ldots, $N_{r}$ based on $M_i$, $M_{i+1}$, \ldots, $M_{r}$. Note that we cannot separately search for augmenting paths in each of the graphs $H_i$, $H_{i+1}$, \ldots, $H_r$ as this would  lead to an algorithm of $\mathcal{O}(r(n+m))$ complexity (matching the complexity of \cite{nimbhorkar2018dynamic}), instead of claimed $\mathcal{O}(\min(c'n ,n^2) + m)$. 

Let us examine two examples depicted in Figure \ref{augfig}. Here the edge $(a_0,p_0)$ is of rank $1$ and the edge $(a_1, p_1)$ of rank $2$.
The vertex $a_1$ belongs to $U(G'_2) \cap E(H'_2)$ and thus is an activated vertex and $p_1$ belongs to $E(G'_2)$ - hence Algorithm \ref{checkingalgorithm} outputs the answer "$M_r$ is not a rank-maximal matching of $H$". This means that in $H'_2$ there exists an $M_2$-augmenting path containing the edges $(a_0,p_0)$ and $(a_1, p_1)$.

In the first example of Figure \ref{augfig}, we can notice that $H'_2$ contains two activated edges of rank $2$ - $(a_1,p_1)$ and $(a_1, p_2)$,
and to obtain a rank-maximal matching of $H_2$, we can augment $M_2$ using either an augmenting path starting at $a_0$, going through  $(a_1,p_1)$ and ending at $p_6$ or a path going through $(a_1, p_2)$  and ending at $p_3$. If $H$ did not contain any edges of rank greater than $2$, then the alternating subgraph $T$ mentioned in the introduction would consist of exactly those two paths. At this point, i.e., in phase $2$ we do not know, however, if at the end of the algorithm - in phase $r$, $T$ will also contain these paths. The only thing we are certain of is that
to obtain a rank-maximal matching of $H$, we have to apply a path beginning at $a_0$, containing $(a_0,p_0), (p_0, a_1)$ and some activated edge incident to $a_1$. Therefore at this point $T$ consists of edges $(a_0,p_0)$ and $(p_0,a_1)$ and we keep an eye on the edges $(a_1, p_1), (a_1, p_2)$.

Next, we  observe that $H'_2$ does not contain any new activated vertices. The graph $H'_3$ is identical to $H'_2$ and to obtain a rank-maximal matching of $H_3$ we may use one of the same two augmenting paths. $T$ does not change.
The vertex $p_1$ belongs to $E(G'_i)$ for every $i$ such that $2 \leq i \leq 4$ but the vertex $p_2$ belongs to $E(G'_i)$ for $i \in \{2,3\}$ and $p_2 \in U(G'_4)$. We can also see  that the graph $G'_4 \cup \{(a_0,p_0), (a_1, p_1), (a_1, p_2)\}$ contains only one  $M_4$-augmenting path.
Thus, if we had augmented $M_2$ using the path going through $(a_1, p_2)$, we would have to change it to get a rank-maximal matching of $H_4$.
On the other hand the path containing $(a_1, p_1)$ was present in the graph $G'_2 \cup \{(a_0,p_0), (a_1, p_1)\}$ and  is still augmenting in the graph 
$G'_4 \cup \{(a_0,p_0), (a_1, p_1)\}$.  We can check that after applying this path we indeed obtain a rank-maximal matching of $H_4$. The subgraph $T$ does not change but we stop observing the edge $(a_1, p_2)$ - we know that eventually this edge will certainly not belong to a rank-maximal matching of $H$. Therefore, to 
be able to update a rank-maximal matching in an efficient way, we observe the endpoints of the activated edges. If there exists an activated edge $e$, whose one endpoint is an activated vertex and the other a vertex of $E(G'_i)$, then we know that to get a  rank-maximal matching of $H_i$, we have to augment a rank-maximal matching of $G_i$.  We do not, however, augment the matching, but continue observing the endpoints. 

In the second example of Figure \ref{augfig}, the vertex $p_1$ belongs to $E(G'_i)$ for every $i$ such that $1 \leq i \leq 4$ and it belongs to $U(G'_5)$. Thus in phases $2-4$ there are no new activated vertices and we use an augmenting path containing $(a_1, p_1)$. The endpoint $p_1$ of the activated edge $(a_1, p_1)$ does not belong to $E(G'_5)$. Hence, $(a_1, p_1)$ ceases to be part of an augmenting path in phase $5$. Indeed, the graph $G'_5 \cup \{(a_0,p_0), (a_1, p_1)\}$ does not contain any augmenting paths and we are stuck with a matching $M_5$ which is not rank-maximal in $H_5$. We observe that if we had augmented $M_4$ in the graph $G'_4 \cup \{(a_0,p_0), (a_1, p_1)\}$ obtaining a rank-maximal matching $N_4$ of $H_4$, then one of the edges of rank $5$ would not be present in the maximum matching of $G'_5 \cup \{(a_0,p_0), (a_1, p_1)\}$ if we computed  it by augmenting $N_4$. So, in order to get a rank-maximal matching of $H_5$ from $M_5$ we should "undo" one of the augmentations that was carried out
in phase $5$. Using matching terminology we should apply any even length $M_5$-alternating path starting at $a$ and  containing $(a_1, p_1)$ and one of the edges of rank $5$ belonging to $M_5$. Observe also that the vertices $a_3, a_4$ belong to $U(G'_5)$ but in $H'_5$ they are part of $E(H'_5)$ - thus we have new activated vertices. Till phase $4$ the alternating subgraph consists of edges $(a_0,p_0), (p_0, a_1)$ and we observe the edge $(a_1, p_1)$. In phase $5$ the subgraph $T$ contains additionally the edges $(a_1, p_1), (p_1, a_2), (a_2, p_2), (p_2, a_3), (a_3, p_3), (p_3, a_4)$ and we observe the activated edges incident to $a_3$ and $a_4$.

\begin{figure}
\centering
    \includegraphics[scale=0.8]{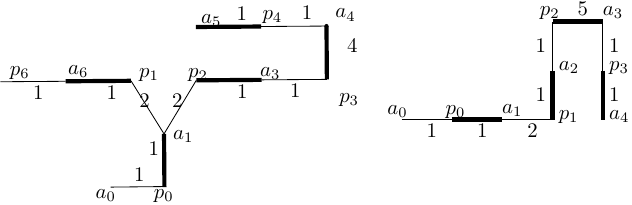}
  \caption{The thick edges belong to the matching.}
\label{augfig}
\end{figure}

To summarize, Algorithm \ref{main} is an extension of Algorithm \ref{checkingalgorithm}, where once we discover in phase $i$ that $M$ is not a rank-maximal matching of $H$, we make a note that the matching will have  to be augmented and start building an alternating subgraph $T$. It turns out that we do not need to augment all of the matchings $M_i$, $M_{i+1}$, \ldots, $M_r$ one by one. Instead we can afford to wait till phase $r$ and then apply either an augmenting path or an even length alternating path contained in the alternating subgraph $T$.  During the course of the computation and before we actually update the matching we keep observing the endpoints of the activated edges. Each activated edge has exactly one endpoint in the subgraph $T$ and forms a potential  extension of $T$. If the endpoint of at least one activated edge belongs to $E(G'_i)$, we are in a so-called augmenting phase. In this phase we do not activate new vertices and do not extend the subgraph $T$. We wait till the endpoints of activated edges fall in $U(G'_i)\cup O(G'_i)$ or the last phase. If the endpoints of activated edges belong to $U(G'_i) \cup O(G'_i)$, we are in a non-augmenting phase,
where we do not need to change the matching  but instead have to activate some new vertices and also extend the subgraph $T$. Augmenting and non-augmenting phases may alternate. An important thing that  allows to save time  is that we do not traverse the graph $H$ beyond the subgraph $T$.

 Once we have $N_r$, we can obtain matchings $N_i$, $N_{i+1}$, \ldots, $N_{r-1}$ easily. We also show how to update reduced graphs $G'_i$ in order to obtain reduced graphs $H'_i$. This part of the algorithm is presented in Section \ref{updates}.

More detailed description of this approach is presented in Section \ref{alg}. In order to prove its correctness we make use of two technical lemmas (Lemma \ref{aux1} and Lemma \ref{biglemma}). Lemma \ref{biglemma} is particularly useful and gives a good characterisation of which vertices need to be activated. It is also crucial for the computation of the Edmonds-Gallai decompositions of the reduced graphs $H'_i$.
\section{Technical Lemmas}
Let $G = (\mathcal{A} \cup \mathcal{P}, \mathcal{E})$ be any bipartite graph. Then $e=(a,p)$ is said to be a {\em new edge for $G$} if $e \notin  \mathcal{E}$
and $a \in \mathcal{A}, p \in \mathcal{P}$. If $M$ is any matching of $G$, then a matching $N$ of $G$ is said to be {\em $M$-augmented maximum} if it is a maximum matching of $G$ obtained by augmenting $M$.

\begin{lemma} \label{aux1}
Let $G = (\mathcal{A} \cup \mathcal{P}, \mathcal{E})$ be a   bipartite graph, $M$ its maximum matching and $C$ a connected component of $G$ that contains exactly one free vertex $a_0$ of $\mathcal{A}$ in $M$.

Let $\mathcal{E}_1=\{(a_1,p_1), (a_2, p_2), \ldots (a_r, p_r)\}$  denote a set of new edges for $G$ such that  each   $a_i$  belongs to $C \cap E(G)$ and no $p_i$ belongs to $C$.  Let $G_1$  denote the graph $G \cup \mathcal{E}_1$ and $n_0=|M|$. Then we have:

\begin{enumerate}
\item  If there exists   $i$  such that $p_i \in E(G)$, then: 
\begin{enumerate}
\item \label{aux1a} Every $M$-augmented maximum matching of $G_1$ contains $n_0$ edges of $\mathcal{E}$ and one edge $(a_i, p_i) \in \mathcal{E}_1$ such that  $p_i \in E(G)$. Conversely, each edge $(a_i, p_i) \in \mathcal{E}_1$ such that  $p_i \in E(G)$ belongs to some  maximum matching of $G_1$. Thus, no edge  $(a_i, p_i) \in \mathcal{E}_1$ such that  $p_i \notin E(G)$ belongs to any $M$-augmented maximum matching of $G_1$.

\item \label{aux1d}(i) Each vertex of $G \setminus C$ either has the same type in $G$ and $G_1$ or (ii) it belongs to $E(G) \cup O(G)$ and $U(G_1)$.
Each vertex of $C$ belongs to $U(G_1)$ or $(E(G) \cap O(G_1)) \cup (O(G) \cap E(G_1))$.

\item \label{aux1e} No edge $(a,p)$ of $G$ such that one of its endpoints belongs to $O(G)$ and the other to $O(G) \cup U(G)$ belongs to any $M$-augmented maximum matching of $G_1$.
 \end{enumerate}

\item If there exists no  $i$ such that $p_i \in E(G)$, then:
\begin{enumerate}
\item  Every  maximum matching of $G$ is also a maximum matching of $G_1$.
Let $P'$ be any even length $M$-alternating path starting at $a_0$. Then $M \oplus P'$ is a maximum matching of $G_1$, which contains $n_0-1 $ edges of $\mathcal{E}$ and one edge of $\mathcal{E}_1$.

\item  Every edge $(a_j, p_j) \in \mathcal{E}_1$ belongs to some even length $M$-alternating path starting at $a_0$.

\item \label{aux2c}(i) Each vertex of $G \setminus C$ either has the same type in $G$ and $G_1$ or (ii) it belongs to $U(G)$ and $E(G_1) \cup O(G_1)$.
Each vertex of $U(G)$ that belongs also to $E(G_1) \cup O(G_1)$ is reachable in $G_1$  from $a_0$ by an even/odd length $M$-alternating path.
Each vertex of $C$  has the same type in $G$ and $G_1$.
\item \label{aux2d} An edge $(a,p)$ such that $a \in U(G) \cap E(G_1)$ and $p \in O(G)$ belongs to some maximum matching of $G_1$. Every other  edge $(a,p)$ of $G$ such that one of its endpoints belongs to $O(G)$ and the other to $O(G) \cup U(G)$ belongs to no maximum matching of $G_1$.
\end{enumerate}
\end{enumerate}
\end{lemma}

\begin{figure}
\centering
  \begin{minipage}[b]{0.4\textwidth}
	\includegraphics[width=\textwidth ]{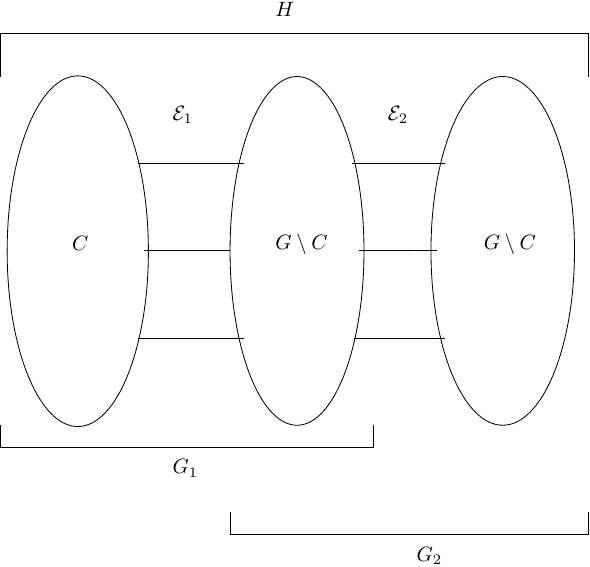}
	\end{minipage}
 \caption{The pictorial representation of $G$, $G_1$, $G_2$ and $H$ as described in Lemma \ref{aux1} and Lemma \ref{biglemma}}
\end{figure}

\begin{proof}
We first prove $1(a)-1(c)$. Let us assume that there exists $i$ such that $p_i \in E(G)$.

From the fact that $a_i, p_i \in E(G)$ and Lemma \ref{EG}, we can see that in $G$ there exists an $M$-alternating path $P_1$ from a free vertex  to $a_i$ and an $M$-alternating path $P_2$ from a free vertex to $p_i$. 
Since $a_i \in C$, the path $P_1$ is contained in $C$ and hence $P_1$ is between $a_0$ and $a_i$. Similarly, the path $P_2$ is contained  in $V \setminus C$. Thus paths $P_1$ and $P_2$ along with the edge $(a_i, p_i)$ form an $M$-augmenting path $P'$ in $G_1$. 
 Let $M' = M \oplus P'$. Obviously $M'$ contains $n_0$ edges of $\mathcal{E}$ and one edge of $\mathcal{E}_1$. 

\textbf{{\em Observation:}} We can note that no $M$-alternating path can contain two edges of $\mathcal{E}_1$. 

This is because all endpoints of $\mathcal{E}_1$ belonging to $C$ are contained in $\mathcal{A}$. Therefore an $M$-alternating path
connecting  two endpoints  of the edges of $\mathcal{E}_1$ in $C$ must have even length. But every edge of $\mathcal{E}_1$ is non-matching. Thus we cannot include any two of them in an $M$-alternating path.

We now prove that $M'$ is a maximum matching of $G_1$. Assume by contradiction that there exists a matching $M''$ such that $|M''| = |M| + 2$. The symmetric difference $M'' \oplus M$ contains two vertex disjoint $M$-augmenting paths in $G_1$. At least one of these paths has both endpoints in $(\mathcal{A} \cup \mathcal{P}) \setminus C$. Let this path be $X$. The path $X$ has to contain at least two edges of $\mathcal{E}_1$ as otherwise it would be contained in the graph $G$, contradicting the maximality of $M$. However, by the observation above we know that it is not possible. By the same observation any $M$-augmenting path contains at most one edge of $\mathcal{E}_1$, which proves $1(a)$.

Note that from the above discussion, it already follows that  each edge $(a_i, p_i) \in \mathcal{E}_1$ such that  $p_i \in E(G)$ belongs to some  maximum matching of $G_1$ and that  no edge  $(a_i, p_i) \in \mathcal{E}_1$ such that  $p_i \notin E(G)$ belongs to any $M$-augmented maximum matching of $G_1$. 


$1(b)$. Let us consider the case $v \in C$ first. It suffices to show that $v \in E(G_1)$ implies $v \in O(G)$ and that $v \in O(G_1)$ implies $v \in E(G)$. Let $v \in O(G_1)$. Let us consider a maximum matching $M_1$ in $G_1$ and any $M_1$-alternating path $P$ from a free vertex $v_0 \in (\mathcal{A} \cup \mathcal{P}) \setminus C$ to $v$ crossing the set $\mathcal{E}_1$. Let $(x, y)$ be the first edge of $\mathcal{E}_1$ on this path and $x \notin C$, $y \in C$. Our goal is to show that $(x, y) \notin M_1$. Assume by contradiction that $(x, y) \in M_1$. Let $P'$ be a subpath of $P$ from a free vertex $v_0$ to $y$. Note that $P'$ has even length and thus $M_1 \oplus P'$ is a maximum matching in $G$ of size $|M| + 1$ - a contradiction. Hence, $(x, y) \notin M_1$ and we have $y \in O(G_1)$. From the definition of $\mathcal{E}_1$, we also have $y \in E(G)$. Since $y \in O(G_1) \cap E(G)$, $v \in O(G_1)$ and there is an alternating path from $y$ to $v$, one can easily see that $v \in E(G)$. Similarly, if $v \in C$ belongs to $E(G_1)$ then it belongs to $O(G)$. This completes the proof of $1(d)$ for the case when $v \in C$.

Suppose next that $v \notin C$. Without loss of generality, we may assume that there exists an $M$-augmenting path $P_1$ in $G_1$  that contains an edge $(a_i,p_i) \in \mathcal{E}_1$ and $p_i$ is free in $M$. This is because, we can apply an even-length $M$-alternating path from a free vertex to $p_i$ and obtain a maximum matching $M'$ of $G$, in which $p_i$ is free.  Let $M_1=M \oplus P_1$. Note that each edge $e \in M_1$ with both endpoints in $V \setminus C$ belongs also to $M$. Let us observe that any $M_1$-alternating path $P'$ that has both endpoints in $V \setminus C$ and crosses $C$ ($V(P') \cap V(C)\neq \emptyset$) contains two edges of $\mathcal{E}_1$, exactly one of which belongs to $M_1$. This follows from the observation in $1(a)$. Suppose that a vertex $v \in V \setminus C$ belongs to $E(G_1)$. There exists then an even-length $M_1$-alternating path $P(v)$ starting at a free vertex $v'$ in $M_1$ and ending at $v$. The free vertex $v'$ must belong to $V \setminus C$. If $P(v)$ does not go through any vertex of $C$, then it is also $M$-alternating and hence $v \in E(G)$. Suppose then that $P(v)$ goes through $C$. It contains then  the edges $(a_i, p_i)$ and $(a_j, p_j)$ of $\mathcal{E}_1$. We observe that the path $P(v)$ must leave the part $V
\setminus C$ via a non-matching edge, hence by $(p_j,a_j)$. This follows from the fact that otherwise it would have to leave $(V \setminus C)$ via $p_i$, but $p_i$ is free in $M$, which would contradict the maximality of $M$.
We split $P(v)$ into three parts $P_1(v), P_2(v), P_3(v)$, where $P_1(v)$ has endpoints $v'$ and $p_j$, $P_2(v)$ - $p_j$ and $p_i$ and $P_3(v)$ - $p_i$ and $v$. Any $M_1$-alternating path with both endpoints in $\mathcal{P}$ (posts) must have even length.
Therefore each of the paths $P_1(v), P_2(v), P_3(v)$ has even length. We notice that $P_3(v)$ is an even length $M$-alternating path from a free vertex $p_i$ to $v$, which means that $v \in E(G)$. The proof for the case when  $v \in V \setminus C$ belongs to $O(G_1)$ is symmetric.

$1(c)$. Assume by contradiction that there exists an edge $(a, p)$ of $G$ such that $a \in O(G)$, $p \in O(G) \cup U(G)$ and $(a, p)$ belongs to a maximum matching $M_1$ of $G_1$. Consider a matching $M' = M_1 \cap \mathcal{E}$. We have $(a, p) \in M'$, $|M'| = |M_1| - 1$, and From $1(a)$, $M'$ is a maximum matching of $G$. However from Lemma \ref{EG} it follows that $(a, p)$ does not belong to any maximum matching of $G$ - a contradiction.

Let us now prove that $2(a)-2(d)$ hold. Assume that there exists no $i$ such that $p_i \in E(G)$.

$2(a)$. We first show that $M$ remains a maximum matching of $G_1$. Assume by contradiction that there exists an $M$-augmenting path $P$ in $G_1$. Let $x$ be the endpoint of the path not belonging to $C$ and let $(a_j, p_j)$ be the first edge of $P$ going from $x$ belonging to $\mathcal{E}_1$. The existence of a subpath from $x$ to $p_j$ implies that $p_j \in E(G)$ - a contradiction. Thus every maximum matching of $G$ is also a maximum matching of $G_1$. The fact that a maximum matching of $G_1$ can contain at most one edge of $\mathcal{E}_1$ follows from the observation in the proof of $1(a)$.

$2(b)$. Let $(a_j, p_j)$ be any edge of $\mathcal{E}_1$. Since $a_j \in E(G)$, there exists an $M$-alternating path $P$ in $G$ from a free vertex $a_0$ to $a_j$. Note that $p_j \notin E(G)$ and Lemma \ref{EG} imply that $p_j$ is matched in $M$ to some $M(p_j) \neq a_j$. The path $P$ together with edges $(a_j, p_j)$ and $(p_j, M(p_j))$ form an even length $M$-alternating path $P'$ in $G_1$. Clearly $M \oplus P'$ is a maximum matching in $G_1$ containing the edge $(a_j, p_j)$, thus $2(b)$ holds.

$2(c)$. Recall that $M$ is a maximum matching in both $G$ and $G_1$, where $G$ is a subgraph of $G_1$. Let $v$ be an even(\emph{resp.} odd) vertex in $G$. There exists an even(\emph{resp.} odd) length $M$-alternating path in $G$ from a free vertex to $v$. Such a path is present in $G_1$, as $G$ is a subgraph of $G_1$. Hence, $v$ is an even (\emph{resp.} odd) vertex in $G_1$. 
Let us have $v \in U(G)$ and $v \in E(G_1) \cup O(G_1)$. From Lemma \ref{EG} there exists an $M$-alternating path $P$ from a free vertex $x$ to $v$. We know that $v \in U(G)$ thus at least one edge of $\mathcal{E}_1$ belongs to $P$. By the observation in the proof of $1(a)$, it follows that exactly one edge of $\mathcal{E}_1$ belongs to $P$. Note that this also means  that $P$ starts from the only free vertex $a_0$ of $C$, thus $2(c)$ holds.

 
$2(d)$. Let $(a,p)$ be an edge such that $a \in U(G) \cup E(G_1)$ and $p \in O(G)$. Hence, $a$ is reachable from a free vertex $a_0$ in $C$ by an even length $M$-alternating path. Therefore, $p$ is reachable from $a_0$ by an odd length $M$-alternating path in $G_1$. If we apply the alternating path, then $(a,p)$ is matched in some maximum matching of $G_1$.

Now let $(a,p)$ be an edge such that $a \in (U(G) \cup O(G)) \cap (O(G_1) \cup U(G_1))$ and $p \in O(G)$. Hence, $p \in O(G_1)$. Thus, the edge $(a,p)$ is an $OO$ or $OU$  edge in $G_1$. Therefore, the edge $(a,p)$ is never matched in any maximum matching of $G_1$. \qed
\end{proof}

We say that a graph $G$ is {\em reduced} if it does not contain any edge  $(u,v)$ such that either both $u$ and $v$ belong to $O(G)$ or exactly one of the vertices belongs to $U(G)$ and the other one to $O(G)$.

Let $G,G_1,G_2, H$ be graphs such that $G=(V,  \mathcal{E}), G_1=(V,  \mathcal{E} \cup \mathcal{E}_1)), G_2=(V,  \mathcal{E} \cup\mathcal{E}_2)$ and $H= (V,  \mathcal{E}\cup\mathcal{E}_1 \cup\mathcal{E}_2)$. Let $M$ be a maximum matching of $G$. 
For $i \in \{1,2\}$, we say that a matching  $M'$ of $H$ is {\em $(M, G_i)$-augmented} if it is obtained by augmenting an $M$-augmented maximum matching of $G_i$.

\begin{lemma}{\label{biglemma}}
Let $G = (\mathcal{A} \cup \mathcal{P}, \mathcal{E})$ be a reduced  bipartite graph and $C$ a connected component of $G$ that contains exactly one free vertex $a_0$ of $\mathcal{A}$ in a maximum matching $M$ of $G$. 

Let $\mathcal{E}_1=\{(a_1,p_1), (a_2, p_2), \ldots (a_r, p_r)\}$ and $\mathcal{E}_2$ denote two sets of new edges for $G$ such that  each  endpoint of  a new edge belongs to $E(G)$. Also, each edge of  $\mathcal{E}_1$ connects a vertex of $C \cap \mathcal{A}$ with a vertex not contained in $C$ and each edge of  $\mathcal{E}_2$ connects two vertices not belonging to $C$.  Let $G_1, G_2$ and $H$  denote respectively $G \cup \mathcal{E}_1, G \cup \mathcal{E}_2$ and $G \cup \mathcal{E}_1 \cup \mathcal{E}_2$.  $M_{12}$ denotes a set of maximum $(M, G_1)$-augmented matchings of $H$  and  $M_{21}$ a set of maximum $(M, G_2)$-augmented matchings of $H$.
 Let $n_0=|M| $ and  $n_2$ denote the number of edges of $\mathcal{E}_2$ contained in a maximum $M$-augmented matching of $G_2$.  Then we have:
	
\begin{enumerate}
\item {\label{biglemma1}} If there exists   $i$  such that $p_i \in E(G_2)$, then
 $M_{12} = M_{21}$ and each matching of  $M_{12}$ contains $n_0$ edges of $\mathcal{E}$, $n_2$ edges of $\mathcal{E}_2$ and one edge of $\mathcal{E}_1$.

\item If there exists no  $i$ such that $p_i \in E(G_2)$, then:
\begin{enumerate}
\item {\label{biglemma4a}} 
A  matching of $M_{12}$ contains $n_0$ edges of $\mathcal{E}$, $n_2-1$ edges of $\mathcal{E}_2$ and one edge of $\mathcal{E}_1$. On the other hand, every   matching of $M_{21}$ is a maximum matching of $G_2$. 

\item {\label{biglemma4b}} 
A  matching of $H$ belongs to $M_{12}$ if and only if it has the form $M' \oplus P'$, where $M' \in M_{21}$ and $P'$ is an even length $M'$-alternating path with one endpoint in $a_0$ and the other in a vertex  $v \in E(G) \cap E(G_1)$, which contains exactly one edge of $\mathcal{E}_1$.

\end{enumerate}
\end{enumerate}

\end{lemma}

\begin{proof}

$(1)$ Let $M_2$ be a maximum matching of $G_2$ and  $(a_i ,p_i) \in \mathcal{E}_1$ an edge with both endpoints in $E(G_2)$. After the addition of  $\mathcal{E}_1$ to $G_2$, $H$ contains  an $M_2$-augmenting path $P'$  containing $(a_i, p_i)$. The path has one edge from $\mathcal{E}_1$, some even length path segments consisting of the edges from $\mathcal{E}$ and some edges from $\mathcal{E}_2$.
Since very augmenting path has an odd number of edges, the number of edges  $\mathcal{E}_2$ in $P'$ must be even. Because every two edges from $\mathcal{E}_2$ are separated by an even length path segment consisting of edges from $\mathcal{E}$, an edge of $M_2 \cap \mathcal{E}_2$is followed  by an edge $\mathcal{E}_2 \setminus M_2$ and vice versa. Hence, half of the edges of $\mathcal{E}_2$ in $P'$ belongs to $M_2$. Therefore, after the application of $P'$ to $M_2$, the number of matched edges of $\mathcal{E}$ and $\mathcal{E}_2$ remains the same. The number of matched edges of $\mathcal{E}_1$ increases to $1$. By the construction every such matching belongs to $M_{21}$ and contains $n_0$ edges of $\mathcal{E}$, $n_2$ edges of $\mathcal{E}_2$ and one edge of $\mathcal{E}_1$.

Next we show that $M_{21}= M_{12}$.  Let $N \in M_{21}$. Consider the symmetric difference $N \oplus M$. It contains $n_2+1$ $M$-augmenting paths, one of which, say $P'$ contains an edge of $\mathcal{E}_1$ and each of the remaining ones one edge of $\mathcal{E}_2$. Notice that
order of applying the $M$-augmenting paths of $M \oplus N$ to $M$ is inconsequential  - we will always get $N$. Since we can first apply $P'$
and afterwards the remaining augmenting paths, it means that $N \in M_{12}$.

Similarly, it can be shown that $M_{12} \subseteq M_{21}$. We  conclude that $M_{12} = M_{21}$.

$2(a)$ There are two ways of obtaining a maximum matching $H$ by augmenting a maximum matching $M$ of $G$. We can first add $\mathcal{E}_2$ to the graph $G$ and augment $M$ in the thus built $G_2$, getting a maximum matching $M_2$ of  $G_2$.  After the addition of  $\mathcal{E}_1$ to $G_2$, there does not exist any edge $(a_i, p_i) \in \mathcal{E}_1$ such that $p_i \in E(G_2)$. Since we do not have any  $M_2$-augmenting path in $H$, the maximum matching of $G_2$ is a maximum matching of $H$.  This shows that any matching of $M_{21}$ is a maximum matching of $G_2$.

 Alternatively, we can first augment $M$ in $G_1$.  Since both endpoints of each edge of $\mathcal{E}_1$ belong to $E(G)$, each such edge  is contained in an  $M$-augmenting path in $G_1$. If we apply any of these $M$-augmenting paths, we get a maximum matching $M_1$ of $G_1$. Note that we can apply only one such path, because $C$ contains only one free vertex in $M$.  $M_1$ contains $1$ edge from $\mathcal{E}_1$ and $n_0$ edges from $\mathcal{E}$. 

Next, we add $\mathcal{E}_2$ and augment $M_1$. We know that none of the vertices of $C$ is free in $M_1$. Therefore there does not exist any $M_1$-augmenting path in $H$ starting from a vertex of $C$. Hence, we have two types of augmenting paths. The first possibility is that the augmenting path is totally contained in $G_2$. The other one is that the augmenting path contains one matched edge and one unmatched edge from $\mathcal{E}_1$. Recall that $C$ can be reached only by edges of $\mathcal{E}_1$. Thus, each $M_1$-augmenting path in $H$ increases the number of the matched edges from $\mathcal{E}_2$ by exactly $1$ and does not change the number of matched edges of $\mathcal{E}_1$ or of $\mathcal{E}$. We already know that the size of a maximum matching of $H$ is $n_0 + n_2$.  Hence,  any matching of $M_{12}$  contains $1$ edge from $\mathcal{E}_1$, $n_0$ edges from $\mathcal{E}$ and $n_2-1$ edges from $\mathcal{E}_2$.

$2(b)$. Consider a symmetric difference of two matchings $N \in M_{12}$ and $M' \in M_{21}$. Since $a_0$ is matched in $N$ and unmatched in $M'$ and because both matchings are maximum in $H$,
$N \oplus M'$  contains an even length $M'$-alternating path $P'$  starting at $a_0$. Observe that in any alternating path or cycle of $N \oplus M'$, any two edges not belonging to $G$ are separated by an even length path segment consisting of edges from $\mathcal{E}$. (Because all edges not contained in $G$ have their endpoints in $E(G)$ and hence all endpoints of such edges incident to one connected component of $G$ are either all contained in $\mathcal{A}$ or all contained in $\mathcal{P}$.) Thus,
any two edges not belonging to $G$ must alternate between edges of $N$ and edges of $M'$ on any alternating path or cycle.
$N \oplus M'$ contains exactly one edge of $\mathcal{E}_1$, which is necessarily contained in $P'$. $P'$ contains also some number of edges
of $\mathcal{E}_2$. 

We claim that there exists a number $2k+1$ such that $P'$ has $k+1$ edges of $N \cap \mathcal{E}_2$ and $k$ edges of 
$M' \cap \mathcal{E}_2$. To prove it, note that any maximal under inclusion alternating path of $M \oplus N$ apart from $P'$ has the same number of edges of $\mathcal{E}_2$ in $M'$ and in $N$. The same applies to any alternating cycle of $N \oplus M'$. On the other hand, we know
that the number of edges of $ \mathcal{E}_2$ in $M'$ is smaller by $1$ than in $N$. Therefore, $P'$ must indeed contain an odd number
of edges of  $\mathcal{E}_2$ and $M' \oplus N$ yields a matching of $M_{12}$. 

Next, we prove that the other endpoint of $P'$ belongs to $E(G) \cap E(G_1)$. Because $P'$ has even length, the other endpoint $a'$ of $P'$
must belong to $\mathcal{A}$. Let $(a_j, p_j)$ denote an edge of  $\mathcal{E}_1$ contained in $P'$. Since $p_j \in E(G)$, $p_j$ is contained in a component of $G$, which contains a free vertex of $\mathcal{P}$ in $M$. The path $P'$ goes between components of $G$, in which a free vertex in $M$ 
belongs alternately to $\mathcal{P}$ and to $\mathcal{A}$. Because $P'$ contains one edge of  $\mathcal{E}_1$ and an odd number of edges of  $\mathcal{E}_2$, 
it ends in a component with a free vertex in $M$ belonging to $\mathcal{A}$. This component is, of course, different from $C$. This shows that $a' \in E(G)$. To see that $a'$ belongs also to $E(G_1)$,
notice that no endpoint of  $\mathcal{E}_2$ belongs to a connected component in $G$ different from $C$ with a free vertex of $\mathcal{A}$ in $M$. 
Therefore for every such component $C' \neq C$ with a free vertex in $\mathcal{A} \cap M$, it holds that any vertex $v \in C'$ has the same type in $G$ and $G_1$.

Conversely, let $P'$ be any even length $M'$-alternating path with one endpoint in $a_0$, the other in $a' \in E(G) \cap E(G_1)$ and containing exactly one edge of  $\mathcal{E}_1$. By the same reasoning as above, we get that $P'$ contains an odd number $2k+1$ of edges of
 $\mathcal{E}_2$, $k+1$ of which belong to $N$. This means that $N \oplus P'$ has $n_0$ edges of  $\mathcal{E}$, one edge of $ \mathcal{E}_1$
and $n_2-1$ edges of  $\mathcal{E}_2$. Also, if we remove from $P'$ edges of $M' \cap  \mathcal{E}_2$, we obtain $n_2$ vertex-disjoint
$M$-augmenting paths in $H$, one of which contains an edge of $\mathcal{E}_1$.
Therefore, $N \oplus P'$ belongs to $M_{12}$. \qed
\end{proof}

\section{Algorithm for Updating a Rank-Maximal Matching} \label{alg}

In this section we present an algorithm for computing a rank-maximal matching of $H$. Its pseudocode is written as Algorithm \ref{main}.
In Algorithm \ref{main} for each $1 \leq i \leq r$, a matching $M_i$ denotes a rank-maximal matching of $G_i$. Also, for each $r \geq j>i$ a matching $M_i$ is contained in $M_j$. 

By phase $i$ of Algorithm \ref{main}, we mean an $i$-th iteration of the loop {\bf for}. By $C_i$ and $R_i$ we denote $C$ or $R$, respectively, at the beginning of phase $i$. By phase $0$ we denote the part of Algorithm \ref{main} before the start of phase $1$.
Depending on whether  during phase $i$ lines 7-13  or 15-19 are carried out,  the   phase is either called  {\em augmenting} or  {\em non-augmenting}.

We say that a vertex $v$ is {\em active} in $G'_i$ if  $v \in O(G'_i) \cup E(G'_i)$ and not active or {\em inactive} (or unreachable) otherwise.

In Algorithm \ref{main} the subgraph $C_i$ contains an "upper" part of the alternating subgraph $T$ mentioned in the introduction and in Section \ref{aug}, i.e., at the end of the algorithm $T$ contains all alternating paths, whose application to $M_r$ results in a rank-maximal matching of $H$. The subgraph $C_i$ may also be viewed as the subgraph that encompasses the ("positive") changes between graphs $H'_i$ and $G'_i$. $C_i$ always contains a new vertex $a_0$ that belongs to $H$ and not to $G$ as well as vertices that are at this point unreachable in $G$, i.e., they belong to $U(G'_i)$. If there exist vertices that are active in $H'_i$ but inactive in $G'_i$, then they belong to $C_i$.
Also, any vertex of $C_i$ is active in some graph $H'_j$ such that $j \leq i$ but inactive in $G'_j$. Any edge that belongs to $H'_i \setminus G'_i$ is either contained in $C_i$ or one of its endpoints belongs to $C_i$. Each such edge belongs to the set $AE$ at some point and is called an {\em activated edge}.
The subgraph $C_i$ does not encompass all changes - in particular, it may happen that some vertex $v \notin C_i$ is active in $G'_i$ but unreachable in $H'_i$
or that some edge belongs to $G'_i \setminus C_i$ but not to $H'_i$.

During phase $i$ we construct a graph $\tilde{H}_i$ that contains every edge belonging to some rank-maximal matching of $H_i$. 
At the beginning of phase $i$, the set $AV$ contains {\em activated vertices}, each of which is alive in $H'_i$ but not alive in $G'_i$.

We later show in Theorem \ref{cor} that a rank-maximal matching of $H_i$ may be obtained from a rank-maximal matching $M_i$ by the application of any alternating path
$s_i \in S_i$,  defined below, and that every matching obtained in this way is rank-maximal. The paths $s_i$ of $S_i$ are defined as follows.

Let $AV_i$ denote the set of activated vertices at the end of  phase $i$.

\begin{definition} \label{sipaths}
For each $i \in \{1,2, \ldots, r\}$ we define the set $S_i$ of $M_i$-alternating paths contained in $\tilde{H}_i$, each of which starts at $a_0$.

If  phase $i$ is augmenting, then 
 $s_i \in S_i$ iff  it is an $M_i$-augmenting path ending at any free vertex in $M_i$. (Each such path contains one edge of $AE$.)

Otherwise, if phase $i$ is non-augmenting, $s_i \in S_i$ iff  (i) it is an $M_i$-alternating path ending at any  vertex of $AV_i$ (each such path contains exactly one edge of $AE_u$ or is a path of length $0$) or (ii) it is an $M_i$-alternating path ending at any  vertex of $Alive(i)$ (each such path contains exactly one edge of $AE_o$). 
\end{definition}

The mentioned earlier alternating subgraph $T$ is formed by paths of $S_r$, i.e., $T= \bigcup_{s_r \in S_r} s_r$. Before phase $r$ (or strictly speaking, before phase $c'$), we are not sure which paths of $S_i$ will eventually belong to $T$, therefore some  paths
of $S_i$ are contained only partially in $C$ and thus also only partially in $T$.


\begin{algorithm} 
\caption{for computing a rank-maximal matching of $H$ }
\label{main}
\begin{algorithmic}[1]
\State $C \leftarrow \{a_0\}$, $AV \leftarrow \{a_0\}$, $AE \leftarrow \emptyset$

\For {$i = 1, 2, \ldots, r$}
	\State $R \leftarrow G'_i \setminus C$
	\State $AE \leftarrow AE \cup \{(a,p)\in  F_i: a\in AV \wedge p \in Alive(i)\}$
	\State $\tilde{H}_i \leftarrow C \cup R \cup AE$

	\If {each $(a,p) \in AE$ is such that $p \in U(G'_i) \cup O(G'_i)$}
		
		\State $AE_u \leftarrow \{(a,p) \in AE: \ p \in U(G'_i)\}$
		
		\ForAll{  even-length $M_i$-alt. path $S$ in $\tilde{H}_i$ between $a_0$ and $a \in U(G'_i) \cap Alive(i)$}
		\State $V(C) \leftarrow V(C) \cup V(S), \ \   \mathcal{E}(C) \leftarrow \mathcal{E}(C) \cup \{(a,p) \in G'_i: a, p \in S\}$
		   \State $AV \leftarrow AV \cup \{a\}$
		\EndFor
		\State $AE_o \leftarrow \{(a,p) \in AE \cup G'_i: a \in C \wedge \ p \in O(G'_i)\}$
		\State $AE \leftarrow AE \cup AE_o \setminus AE_u$
	  \Else \State (there exists $(a,p) \in AE$  such that $p \in E(G'_i)$)
	  \ForAll{ $(a,p) \in AE$ such that $p \in O(G'_i) \cup U(G'_i)$ }
	  		\State $AE \leftarrow AE \setminus \{(a,p)\}$
	  \EndFor	
	  \State $AV \leftarrow \emptyset$	
     \EndIf
\EndFor
\State return $M_r \oplus s^*_r$, where $s^*_r$ - a path of $S_r$ of minimal length
\end{algorithmic}
\end{algorithm}

In order to obtain a rank-maximal matching $N_r$ that differs from $M_r$ in the smallest possible way, we choose that path $s_r \in S_r$,  which is shortest. \\

\begin{theorem} \label{cor}
For each $i \in \{1,2, \ldots, r\}$ it holds:
\begin{enumerate}

\item $\tilde{H}_i$ contains every edge belonging to some rank-maximal matching of $H_i$.
\item $C_i$ has the properties of $C$ from Lemma \ref{aux1} with respect to the matching $M_i$. $C_i$ contains no vertex that is active
in $G'_i$.  After the execution of line $4$ of phase $i$ each edge of $AE$ connects a vertex of $C_i$ with a vertex of $R_i$.
\item \label{cor3} For each $s_i \in S_i$ a matching $M_i \oplus s_i$ is a rank-maximal matching of $H_i$ and a maximum matching of $\tilde{H}_i$.
\item At the end of phase $i$ the set $AV$ consists of all vertices that are alive in $H'_{i+1}$ but not alive in $G'_{i+1}$.
\end{enumerate}

\end{theorem}
\begin{proof} 
{\em Point $1$.} At the end of phase $0$, the subgraph  $C_0$ as well as the set $AV$ contains exactly one vertex $a_0$. We can notice that $H_1=\tilde{H}_1 = H'_1$. Using Lemma \ref{aux1} it is easy to check that  the theorem holds  for $i=1$.

Suppose now that $i>1$. We first argue that every edge of $H'_i \setminus H'_{i-1}$ is present in $\tilde{H}_i$. By Fact \ref{alive} every edge of $H'_i \setminus H'_{i-1}$ has both endpoints in the set of alive vertices of $H'_i$. If both endpoints of such an edge belong to $R_i$, then they are also alive in $G'_i$ (a vertex that is not alive in $G'_i \setminus C_i $ is also not alive in $H'_i$) and thus such an edge is included in $\tilde{H}_i$. If one of the endpoints $v$ of such edge $(v,w)$ belongs to $C_i$, then $v$ must belong to the set $AV$ - by the induction hypothesis point 4 and hence edge $(v,w)$ is added to $\tilde{H}_i$ in line $4$.

The edges belonging to $\tilde{H}_{i-1} \setminus \tilde{H}_i$ are either those belonging to $G'_{i-1} \setminus G'_i$ or those removed in line $17$ during phase $i-1$.  By Lemma \ref{aux1} 1(b), 1(c) and 2(d), no such edge belongs to a maximum matching of $\tilde{H}_{i-1}$ and therefore by the induction hypothesis point 3, no such edge belongs to a rank-maximal matching of $H_{i-1}$. 

We have thus proved that $\tilde{H}_i$ contains every edge belonging to some rank-maximal matching of $H_i$.

{\em Point $2$}. It is easy to see that $C_i$ satisfies all the properties stated in point $2$ of the theorem.

{\em Point $3$}. We will now show  that every matching $M_i \oplus s_i$ is a maximum matching of $\tilde{H}_i$ that contains a rank-maximal matching of $H_{i-1}$. We are going to make use of the following claim.
	
\textbf{Claim:}
Suppose that $i \geq 1$ and $\tilde{H}_i$ contains every edge belonging to some rank-maximal matching of $H_i$. Then a maximum matching of $\tilde{H}_i$ that contains a rank-maximal matching of $H_{i-1}$, is a rank-maximal matching of $H_i$. 

{\em Case 1:} Phases $i-1$ and $i$ are  augmenting. There exists then an edge  $(a,p) \in AE$  such that $p \in E(G'_i)$. In phase $i-1$ the edge $(a,p)$ also belongs to $AE$, because in phase $i$ we do not add any new edges to $AE$ since $AV=\emptyset$. We claim that $p \in E(G'_{i-1})$.  If it were the case that $p \notin E(G'_{i-1})$, then  $(a,p)$ would have been deleted during phase $i-1$ in line $17$. Since $(a,p)$ is such that $p \in E(G'_i)$,  there exists in $\tilde{H}_{i}$ an $M_{i}$-augmenting path $T$ containing $(a,p)$ that ends at some free vertex $p'$ in $G'_{i} \setminus C_{i}$. The vertex $p'$ is also free in $G'_{i-1} \setminus C_{i-1}$. If $T$ is contained in  $\tilde{H}_{i-1}$, then by the definition of the set $S_{i-1}$ the path $T$ belongs to $S_{i-1}$ and thus by the induction hypothesis, $M_{i-1} \oplus T$ is a rank maximal matching of $H_{i-1}$. This then means that $M_i \oplus T$ is a rank-maximal matching of $H_i$, because $M_i \oplus T$ contains the same number of rank $i$ edges as $M_i$. 

Next assume that the path $T$ is not contained in $\tilde{H}_i$ (Figure \ref{correctness}(a)). Let $T'$ denote the maximal subpath of $T$ that starts at $a_0$ and is contained in $\tilde{H}_{i-1}$. It must end at a vertex $p''$ that is free in $M_{i-1}$ and  matched in $M_i$. We know that $M_{i-1} \oplus T'$ is a rank-maximal matching of $H_{i-1}$. The path $T'' = T \setminus T'$ connects two vertices that are alive in $G'_{i+1}$, one of which is free in $M_i$. Also, $T''$ is contained in $G'_i \setminus C_i$. Theorem 1 of \cite{GhoshalNN14}
states that given a rank-maximal matching $M_i$ of $G_i$ and an even-length $M_i$-alternating path $T''$ with one endpoint free in $M_i$ and the other alive in $G_{i+1}$, the matching $M_i \oplus T''$ is rank-maximal in $G_i$.
Thus, the matching $M_i \oplus T''$ is rank-maximal  in $G_i$, which means that $M_i \oplus T$ is a rank-maximal matching of $H_i$.

 {\em Case 2:} Phase $i-1$ is augmenting and $i$   non-augmenting.  Let $n_0$ denote the number of edges of $M_i$ that have rank smaller than $i$ and $n_2$ the number of edges of $M_i$ with rank $i$. Let us note, that since phase $i-1$ is augmenting, each matching $M_{i-1} \oplus s_{i-1}$ contains one edge more than $M_{i-1}$. A maximum matching of $\tilde{H}_i$ has the same cardinality as $M_i$. Therefore, a rank-maximal matching of $H_{i}$ contains at most $n_2-1$ edges of rank $i$.

Each $s_i$ from $S_i$ is an even length $M_i$-alternating path connecting $a_0$ and a vertex $a \in AV_i \cup Alive(i+1)$. By Lemma \ref{biglemma} 2(a),2(b)  the  matching $M_i \oplus s_i$ is a maximum matching of $\tilde{H}_i$ that has one edge of $AE$, $n_0$ edges not belonging to $AE$ and of rank strictly smaller than $i$ and $n_2-1$ edges of rank $i$. Each of the paths  $s_i \in S_i$ contains some path  $s_{i-1} \in S_{i-1}$. Also, each edge of $AE$ belongs to $S_i \cap S_{i-1}$. It is so because each edge $(a,p)$ that belongs to $AE$ at the beginning of phase $i$ is such that $p \in E(G'_{i-1})$ and thus $(a,p)$ is contained in some $s_{i-1}$ as well as some $s_i$. Hence, $M_i \oplus s_i$ is a rank-maximal matching of $H_i$.


{\em Case 3:} Suppose now that phase $i-1$ is non-augmenting and phase $i$  augmenting (Figure \ref{correctness}(b)). It means that there exists an edge $e=(a,p) \in AE$ such that $p \in E(G'_i)$. If the edge $e$ did not belong to $AE$ in phase $i-1$, it means that $a \in AV_{i-1}$ and thus there exists a path $s_{i-1}$ ending at $a$, whose application to $M_{i-1}$ yields a rank-maximal matching of $H_{i-1}$. Also, in $G'_i \setminus C_i$ there exists an even length $M_i$-alternating path $T$ that starts at $p \in Alive(i+1)$ and ends at a free vertex $p'$. By Theorem 1 of \cite{GhoshalNN14}, the matching $M_i \oplus T$ is rank-maximal in $G_i$.
The edge $e$ clearly has rank $i$. Therefore $M_i \oplus (\{e\} \cup T \cup s_{i-1})$  is a maximum matching of $\tilde{H}_i$ that contains a rank-maximal matching of $H_{i-1}$ and has one edge of rank $i$ more than $M_i$ - therefore it is rank-maximal in $H_i$. 

If the edge $e$ did belong to $AE$ in phase $i-1$, then $p \in O(G'_{i-1})$. In $G'_i \setminus C_i$ there exists an even length $M_i$-alternating path $T$ that  starts at $p$ and ends at a free vertex $p'$ (Figure \ref{correctness}(c)). Let $T'$ denote the maximal subpath of $T$ starting at $p$ and contained in $G'_{i-1} \setminus C_i$ and let $T''$ denote $T \setminus T'$. Also, let $s$ denote an $M_i$-alternating path from $a_0$ to $p$.
We notice that $T'$ ends at a vertex $a''$ alive in $G'_{i}$ and thus $M_{i-1} \oplus (s \cup T')$ is rank-maximal in $H_{i-1}$. Let us note that the edge $e'=(p'', a'')$ is of rank $i$. Also, $M_i \oplus (T'' \setminus e'')$ is rank-maximal in $G_i$. Therefore $M_i \oplus (s \cup T)$ is rank-maximal in $H_i$.

{\em Case 4:} In the final case, we assume that phases $i-1$ and $i$ are both non-augmenting. Hence a maximum matching of $\tilde{H}_i$ has the same cardinality as $M_i$. 
Each of the paths $s_i \in S_i$ contains some path of $S_{i-1}$. Let  $s_{i-1}$ denote a maximal subpath of $s_i$ that belongs to $S_{i-1}$.  First we prove that $M_i \oplus s_{i-1}$ is a rank-maximal matching of $H_i$. Since phase $i-1$ is  non-augmenting, $s_{i-1}$ is  an even length alternating path from $a_0$ to some $a' \in AV_{i-1} \cup Alive(i)$ and $M_{i-1} \oplus s_{i-1}$ yields a rank-maximal matching of $H_{i-1}$. Therefore $M_i \oplus s_{i-1}$ is also a maximum matching of $\tilde{H}_i$ that contains a rank-maximal matching of $H_{i-1}$. Thus $M_i \oplus s_{i-1}$ is a rank-maximal matching of $H_i$.

Let $T = s_i \setminus s_{i-1}$. 
Since $s_{i-1}$ is a maximal subpath belonging to $S_{i-1}$, the edge in $T$ incident to $a'$ is of rank $i$ and $T$ does not contain any edge from $C_{i-1}$. Therefore $T \cap G'_{i-1}$ is a collection of path segments contained in $G'_{i-1} \setminus C_{i-1}$ and $T$ is obtained by connecting  such segments with rank $i$ edges. The endpoints of these path segments  belong to $Alive(i)$. In other words,  $T \cap G'_{i-1}$ is a collection of even length alternating paths in $G'_{i-1} \setminus C_{i-1}$ from a free vertex to an alive vertex. By Theorem $1$ of \cite{GhoshalNN14}, $(T \cap G'_{i-1}) \oplus M_{i-1}$ is a rank-maximal matching of $G_{i-1}$ and consequently a rank-maximal matching of $H_{i-1}$. Thus  $M_i \oplus s_i$ is a maximum matching of $\tilde{H}_i$ and contains a rank-maximal matching of $H_{i-1}$. Therefore, $M_i \oplus s_i$ is a rank-maximal matching of $H_i$.

{\em Point $4$}. Before  proving point $4$, let us note the following relationships between EG-decompositions of $\tilde{H}_i$ and $H'_i$:
(i) if $v \in E(H'_i)$, then  $v \in E(\tilde{H}_i$), (ii) if $v \in O(H'_i)$, then  $v \in O(\tilde{H}_i$).
They follow from the fact that each $EO$-edge of $H'_i$ belongs to some rank-maximal matching of $H_i$ and hence, by point $1$ of the current theorem it also belongs to $\tilde{H}_i$. The corollary of these two implications is:
(iii) if $v \in U(\tilde{H}_i)$, then $v \in  U(H'_i)$.

If phase $i$ is augmenting, then we set the set $AV$ as empty. This is because  by point $1(b)$ of Lemma \ref{aux1},
any vertex $v$ that  belongs to $C$ changes type.  On the other hand any vertex that belongs to $R$ has the same type in $G'_i$ and $\tilde{H}_i$. Therefore, it cannot belong to $AV_i$ either.

Suppose now that phase $i$ is non-augmenting. We first prove that every vertex  $v \in AV_i$ is alive in $H'_{i+1}$ but not in $G'_{i+1}$. We have two possibilities. First, assume that $AV$ contains  $v$ at the end of phase $i-1$. 
 
By the induction hypothesis, at the end of phase $i-1$, $v$ is alive in ${H}'_{i}$ but not  in $G'_{i}$. It means that it was added to 
$AV$ during some phase $j<i$. In phase $j$ there also existed an even-length $M_j$-alternating  path $P'$ from $a_0$ to $v$. This path was added then to $C$. Thus the path $P'$ is also present in $\tilde{H}_i$ (because no edge or veretx  is removed from $C$) and it belongs to $S_i$.
By point $3$ of the current theorem  $M_i \oplus P'$ is a rank-maximal matching of $H_i$, in which $v$ is free. It means that $v \in E(H'_i)$.
Therefore, $v$ is indeed alive  in $H'_{i+1}$ but not in $G'_{i+1}$.

On the other hand, if $v$ is added to $AV$ during phase $i$, then it is an endpoint of some path $s_i \in S_i$ and by point $3$ of the current lemma, $M_i \oplus s_i$ is a rank-maximal matching of $H_i$. Thus $v$ is free in the rank-maximal matching $M_i \oplus s_i$ of $H_i$, which means that is is also free in some rank-maximal matching of $H_j$ for each $j<i$. Therefore, $v$ is alive in $H'_i$ and since $v \in U(G'_i)$, it is not alive in $G'_i$.

 Conversely, suppose that $v$ is alive in $H'_{i+1}$ but not alive in $G'_{i+1}$. Then there is an even length alternating path from a free vertex to the alive vertex $v$ in $H'_i$. Now, by Theorem $1$ of \cite{GhoshalNN14}, every edge of the path belongs to some rank-maximal matching of $H_i$. Hence the whole path is present in $\tilde{H}_i$ and $v \in E(\tilde{H}_i)$. Since $v$ is alive in $H'_{i+1}$, $v$ is also an alive vertex in $\tilde{H}_{i+1}$. If $v$ is not alive in $G'_{i}$, then by the induction hypothesis, $v$ is added to $AV$ during phase $i-1$ and since phase $i$ is non-augmenting, it also belongs to $AV$ in phase $i$. Suppose now that $v$ is alive in $G'_{i}$. This implies that $v \in R$ after phase $i-1$. Since  phase $i$ is  non-augmenting and $v$ is an alive vertex in $\tilde{H}_{i+1}$ but not in $G'_{i+1}$, by point $2(c)$ of Lemma \ref{aux1}, $v$ must belong to  $U(G'_i) \cap E(\tilde{H}_i)$. There is an even length $M_i$-alternating path $P'$ from a free vertex to $v$ in $\tilde{H}_i$ but there is no such path in $G'_i$. Therefore the path $P'$ must contain at least one activated edge. No activated edge belongs to $M_i$ in $\tilde{H}_i$. Once $P'$ enters $C$ using that activated edge, the path can't leave $C$ and $a_0$ is the only free vertex inside $C$. Hence, there is an even length alternating path from $a_0$ to $v$ in $\tilde{H}_i$. Finally, by lines $8$ and $10$ of Algorithm \ref{main}, $v$ is added to $AV$ during phase $i$. \qed
\end{proof}

\begin{figure}
  \centering  
  \begin{minipage}[b]{\textwidth}
   \includegraphics[width=\textwidth]{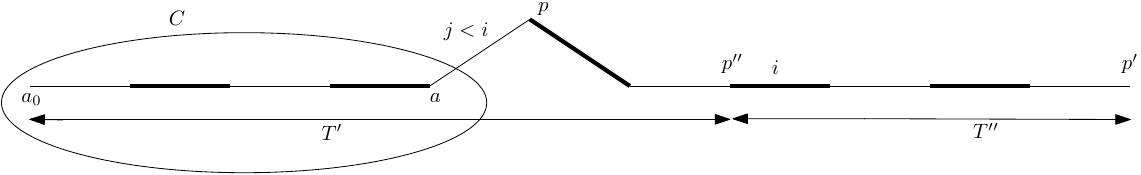}
    \centerline{} 
  \end{minipage} 
  (a)
\hspace{100mm}
  \begin{minipage}[b]{\textwidth}
    \includegraphics[width=\textwidth]{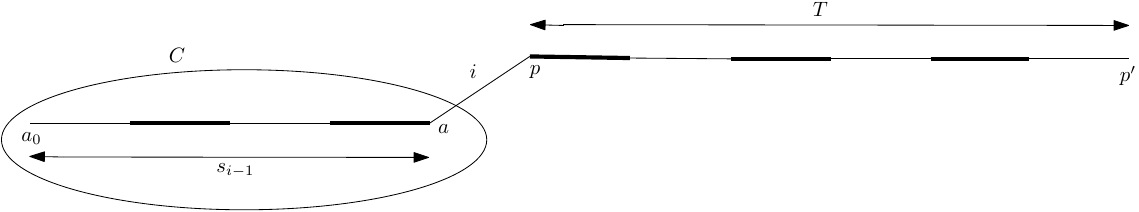}
    \centerline{} 
  \end{minipage}
  (b)
   \begin{minipage}[b]{\textwidth}
    \includegraphics[width=\textwidth]{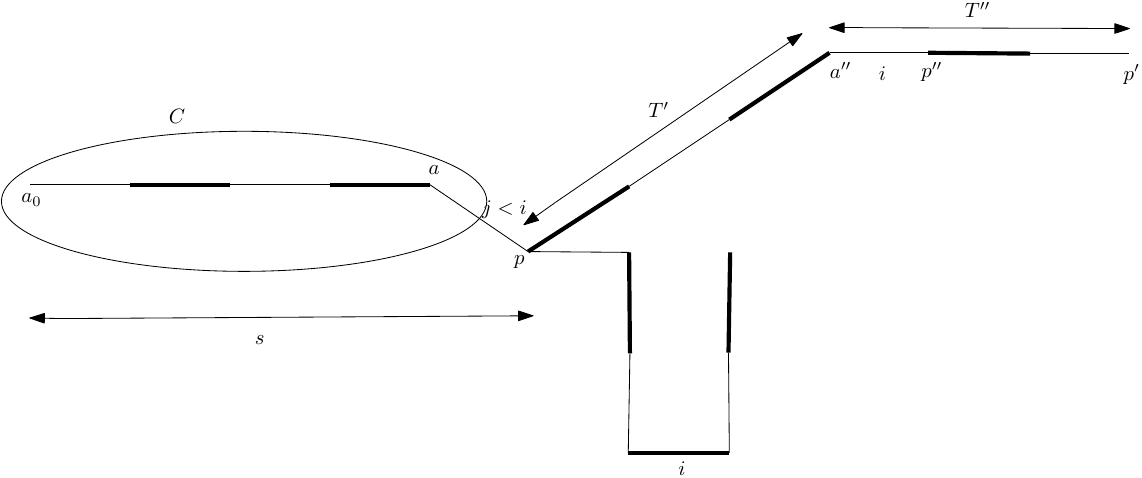}
    
  \end{minipage}
  (c)
  \caption{The figures represent different cases we encounter during the correctness proof of the algorithm. (a-b) The first two figures represent the case when both iterations $i-1$ and $i$ are  augmenting. (c) The third figure represents the case when the iteration $i-1$ is   non-augmenting  but $i$   augmenting.}
	
\label{correctness}
\end{figure}

\begin{theorem} \label{cor2}
Algorithm \ref{main} runs in $\mathcal{O}(\min(c'n ,n^2) + m)$ time.
\end{theorem}

\begin{proof} 
Without any additional assumptions the running time of Algorithm \ref{main} is $\mathcal{O}(rn + m)$. The extension of $C$ to include $M_i$-alternating paths may be easily implemented to take $\mathcal{O}(m)$ time in total as the set $C$ does not shrink. 
Also, during the process of the extension of $C$ to include $M_i$-alternating paths, we only traverse  $R$ starting from the activated edges of type $(a_i, p_i)$, where $a_i \in C$ and $p_i \in U(G'_i)$. We  never reach $C$ while traversing  these paths and each such traversal ends with the inclusion of some path to $C$ in after that phase.  Hence we do not traverse any edge inside $C$ or inside $R$ more than once. 

The executions of line 8 and 15 may require $\mathcal{O}(rn)$ time as every edge $(a,p)$ may belong to $AE$ for a number of phases and in each one of them we need to check to which of the sets $E(G'_i), U(G'_i), O(G'_i)$ the endpoint $p$ belongs. The rest of the time  the algorithm takes is $\mathcal{O}(m)$.

Now we show how to store the reduced graphs after every phase and their GE-decomposition efficiently. Let us consider a vertex $v \in V(G)$. Before the first iteration $v$ is an even vertex. During the algorithm, the type of $v$ switches between an even vertex and an odd vertex. If $v \in U(G'_i)$ for some iteration $1\leq i \leq r$, then $v$ remains an unreachable vertex for every subsequent iterations. Hence we maintain an ordered list $i_1 < i_2 < \ldots < i_k$, which denotes the phases when $v$ changes its type, i.e., $v$ belongs to $O(G'_{i_1}), E(G'_{i_2}), \ldots, U(G'_{i_k})$. Each such list has length at most $n$. Combining these lists, we obtain a list of lists  of size $\mathcal{O}(n^2)$. This list of lists stores the GE-decomposition of every reduced graphs $G'_1, G'_2, ..., G'_r$. 
 
For any edge $(a,p)$ in $G$, its presence in each reduced graph can be computed from the GE-decomposition of the end-points of that edge. Hence we do not need to store any extra information about the edges of $G$ to reconstruct the reduced graphs. 

Also, similarly as in the case of the algorithm for a rank-maximal matching, we may modify the algorithm so that it runs in $\mathcal{O}(\min(c'n ,n^2) + m)$ time, where $c'$ denotes the maximal rank in an optimal solution. To this end, it suffices to stop when there are no new edges to add. \qed
\end{proof}

\begin{theorem}
Assuming we are given the reduced graphs $G'_1, G'_2, \ldots, G'_r$, the reduced graphs of $H$ can be computed in $\mathcal{O}(\min(c'n ,n^2) + m)$ time.
\end{theorem}



\section{Example of How Algorithm \ref{main} Works} \label{example}

\begin{figure}
\centering
   \includegraphics[width=.8\textwidth]{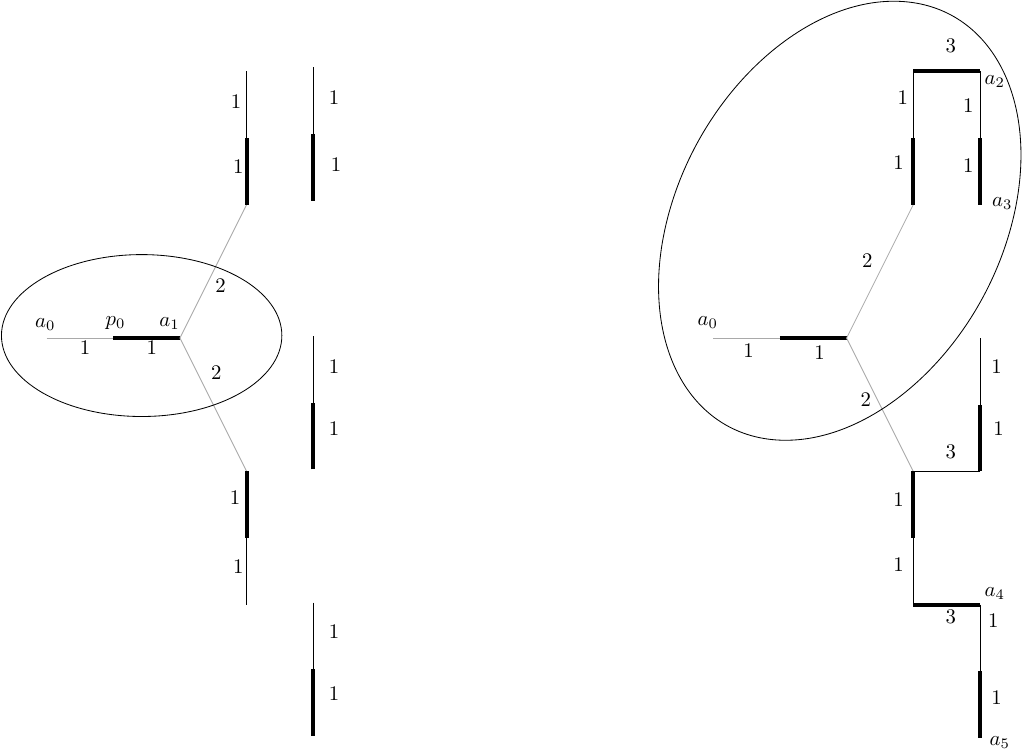}
	\caption{On the right-hand side there is the graph $H$ obtained by adding the vertex $a_0$ and the edge $(a_0, p_0)$ to $G$. The graph $\tilde{H}_2$ is presented on the left-hand side. In this particular example the graph on the right-hand side is also equal to $\tilde{H}_3$ (i.e. $H = \tilde{H}_3$). The vertices inside the ellipse form the set $C$. Label of each edge is equal to its rank.}
	\label{algofig}
\end{figure}

Let us take a look at what Algorithm \ref{main} does when executed on the graph presented on Figure \ref{algofig}. At the beginning of the first iteration, $C$ contains only vertex $a_0$, $R$ contains the rest of the graph. The vertex $a_0$ is added to the set $AV$. The edge $(a_0, p_0)$ is added to the set $AE$ of activated edges. Since $p_0 \in U(G'_1)$, sets $C$ and $AV$ are updated. The set $C$ from now on contains the subgraph inside the ellipse on the left-hand side of Figure \ref{algofig} and $AV$ contains the vertices $a_0$ and $a_1$.  In order to get a rank maximal matching of $H_1$,
we may apply one of the two  alternating paths - each one  starts at $a_0$ and one of them finishes at $a_1$ and the other at $a_0$ (a zero-length path).

In the second iteration the algorithm first updates the set of activated edges. The updated graph is presented on the left-hand side of Figure \ref{algofig}. At this point the set $AE$ contains two edges incident to $a_1$  and crossing the ellipse. Since both endpoints of edges of $AE$ at this point belong to $E(\tilde{H}_2)$, the algorithm enters an augmenting phase. The set $AV$ is reset to empty. In order to obtain a rank-maximal matching of $H_2$ any augmenting path of $\tilde{H}_2$ may be applied. Sets $C$ and $R$ remain the same.

In the third iteration the algorithm first updates the sets $AV$ and $C$. The set $AV$ contains the vertices $a_2$ and $a_3$. The updated graph is presented on the right-hand side of Figure \ref{algofig}. Since there are no edges in $AE$ such that both endpoints belong to $E(\tilde{H}_3)$, the algorithm enters a non-augmenting phase. The set $AE$ contains an edge of rank $2$ crossing the ellipse. Vertices which are alive and reachable from $a_0$ after this iteration are denoted as $a_4$ and $a_5$. In order to obtain a rank-maximal matching of $H_3$ we can either apply an even length alternating path starting at $a_0$ and ending at an alive vertex or apply an even length alternating path starting at $a_0$ and ending at a vertex of $AV$.

\begin{figure}
\centering
   \includegraphics[width=.5\textwidth]{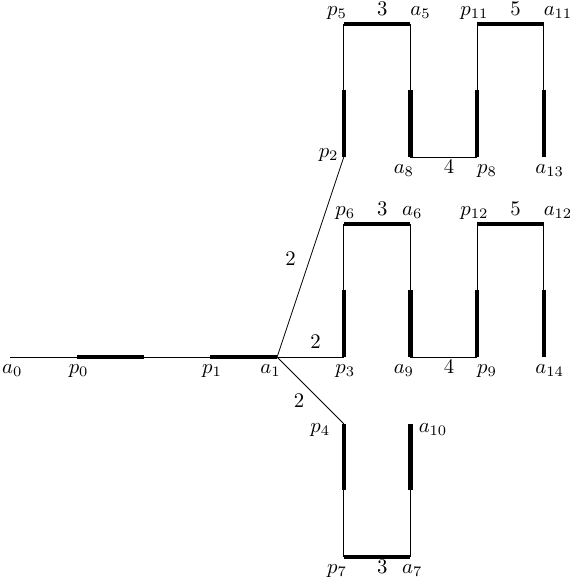}
	\caption{The figure presents the graph $H$ which is obtained by adding the edge $(a_0, p_0)$ to $G$. Thick edges belong to a rank-maximal matching of $G$.}
\label{algofig2}
\end{figure}

Let us now take a look at the example presented in Figure \ref{algofig2}. This example gives us an idea about changing from an augmenting phase to a non-augmenting phase and vice versa. Before the first iteration, the vertex $a_0$ is an activated vertex and $(a_0, p_0)$ is an activated edge and the algorithm enters a non-augmenting phase after the first iteration. In the second iteration the edges $(a_1, p_2)$, $(a_1, p_3)$, $(a_1, p_4)$ of rank $2$ become activated. Also $p_2$, $p_3$ and $p_4$ are the even vertices in $G'_2$, thus the algorithm enters an augmenting phase and a rank-maximum matching in $H_2$ can be obtained by applying any of the augmenting paths starting from $a_0$. The vertices $p_2$, $p_3$ and $p_4$ change their type to unreachable in $G'_3$. Therefore, the algorithm changes to a non-augmenting phase.  

In iteration $4$, $a_8$ and $a_9$ are some of the activated vertices present in the graph. $(a_8, p_8)$ and $(a_9, p_9)$ are two activated edges of rank $4$. Since both $p_8$ and $p_9$ belong to $E(G'_4)$,  the phase changes to an augmenting phase. After phase $5$, there is no augmenting path starting from $a_0$ present in the graph. Therefore the graph enters to a non-augmenting phase after the final iteration.  After the final phase, $a_{11}$, $a_{12}$, $a_{13}$ and $a_{14}$ are the activated vertices present in the graph.  A rank-maximal matching of $H$ can be computed by applying an even length alternating path from $a_0$ to any of the activated vertices.

\section{Updates of Reduced Graphs $H'_i$} \label{updates}

Algorithm \ref{main} computes graphs $\tilde{H}_i$ such that $\tilde{H}_i$ contains every edge that belongs to some rank-maximal matching of $H_i$. Our goal now is to determine which edges of $\tilde{H}_i$ do not belong to  $H'_i$. 



By Theorem \ref{cor} point $3$ for any $s_i \in S_i$ matching $M_i \oplus s_i$ is rank-maximal in $H_i$.
For each $i$ we choose an arbitrary $s_i \in S_i$ and denote the matching $M_i \oplus s_i$ as $N_i$.

\begin{lemma} \label{hgDec}
A vertex $v$ belongs to $E(H'_i)$ (resp. $O(H'_i)$) iff in the graph $\tilde{H}_i$ there exists an $N_i$-alternating path $P'$ from a free vertex $u$ to a vertex $w\in Alive(i+1) \cup AV_i$ such that $v$ lies on $P'$ and its distance on $P'$ from $u$ is even (respectively, odd).
\end{lemma}

\begin{proof}
Suppose first that $v \in E(H'_i)$. Then in $H'_i$ there exists an $N_i$-alternating path $P'$ from a free  vertex $u$ to a vertex $w \in \bigcap_{j=1}^{i} E(H'_j)$ such that $v$ lies on $P'$ and its distance on $P'$ from $u$ is even. By Theorem $1$ of \cite{GhoshalNN14} $N_i \oplus P'$ is also a rank-maximal matching of $H_i$. This means that every edge of $P'$ belongs to some rank-maximal matching of $H_i$ and therefore by Theorem \ref{cor} point 1 the graph $\tilde{H}_i$ contains each edge of $P'$.  By Theorem \ref{cor} point 4 we have  $\bigcap_{j=1}^{i} E(H'_j)=Alive(i+1) \cup AV_i$.

The case when $v \in O(H'_i)$ is analogous.

Let us assume now that in the graph $\tilde{H}_i$ there exists an $N_i$-alternating path $P'$ from a free vertex $u$ to a vertex $w\in Alive(i+1) \cup AV_i$ such that $v$ lies on $P'$ and its distance on $P'$ from $u$ is even. It suffices to show that every edge of $P'$ belongs to some rank-maximal matching of $H_i$, because it will mean that every edge of $P'$ belongs to $H'_i$.

If every edge of $N_i \cap P'$ belongs also to $M_i$, then the whole path $P'$ lies in $G'_i \setminus C_i$ and thus by Theorem $1$ from \cite{GhoshalNN14} $N_i \oplus P'$ is also a rank-maximal matching of $H_i$ and we are done.

Next assume that $P'$ is not totally contained in $G'_i \setminus C_i$. In other words, $V(P' \cap C_i) \neq \phi$. Let $u$ denote an endpoint of $s_i$. Clearly, $u \in Alive(i+1) \cup AV_i$.

Assume that  phase $i$ is non-augmenting. Then  by Definition \ref{sipaths}  every  path of $S_i$ ends with a vertex matched in $M_i$. After the application of $s_i$, $u$ is the only free vertex that is reachable from $C$ by an $N_i$-alternating path. Therefore, the free vertex in $P'$ must be $u$. Let the other endpoint of $P'$ be $w$.

Let $x$ be the last  vertex on  $P'$ (considered from $u$ to $w$) that belongs to  $s_i$. It must hold that $x \in C_i$ - otherwise, $P'$ is contained in $G'_i \setminus C_i$. 

Let $P_2$ denote the subpath of $P'$ between $x$ and $w$ and $s^1_i$  the subpaths of $s_i$ between $a_0$ and $x$.   We can notice that the path consisting of two subpaths $s^1_i$ and $P_2$ forms a path $s'_i$ that belongs to $S_i$. Therefore, $M_i \oplus s'_i$ is a rank-maximal matching of $H_i$, in which $w$ is free. A path consisting of a single vertex $a_0$ also belongs to $S_i$, hence every edge of $M_i$ belongs to $H'_i$. This means that every edge of $s'_i$ and hence every edge of $P_2$ belongs to $H'_i$.

When we look at the symmetric difference of $P'$ and $s_i$, it consists of $s'_i$ and possibly some number of cycles. The cycles of $P'\oplus s_i$ and the path $s'_i$ are vertex-disjoint, because every vertex on $s_i$ is matched both in $M_i$ and $N_i$ with an edge belonging to $s_i$. This means that each vertex on $P'$ is matched via an $N_i$-edge, which does not belong to $P' \oplus s_i$ (because such an edge belongs both to $P'$ and $s_i$). For the same  reason each cycle $C'$ contained in $P' \oplus s_i$  is $M_i$-alternating.  To complete the proof it suffices to show that every edge of $C'$  belongs to some rank-maximal matching of $H_i$.  

\textbf{Claim:} Any $M_i$-alternating cycle $C'$   in $\tilde{H_i}$ is also present in $G'_i$.  

\begin{proof}  If $C'$ is not contained in $G'_i$, then it must contain some activated edges. No activated edge belongs to $M_i$. We label each vertex $v$ of $C'$ as follows: we give it a label $R$ if it does not belong to $C$ at the beginning of phase $i$; otherwise, we give $v$ a label $C_{j}$ if it was added to $C$ during phase $j<i$. Let $e=(a,p)$ be any activated edge contained in $C'$. No activated edge belongs to $M_i$. Then, either (i) $a \in C_{j}, p \in C_{j'}$ for some $j<j'<i$ or (ii)  $a \in C_{j}, p \in R$.  Let $k$ be a minimum index such that some vertex of $C'$ belongs to $C_k$. Recall that no activated edge belongs to $M_i$.
   Hence, the part of $C'$ contained in $C_k$ would have to be of odd length (compare the observation in the proof of Lemma \ref{aux1}).
But this means that we have arrived at a contradiction, because all activated edges of $C'$ incident to $C_k$ are incident to vertices of $\mathcal{A}$. This means that $C'$ is present  in $G'_i$.  \qed
\end{proof}


By Theorem 1 of \cite{GhoshalNN14}, every edge of $C'$  belongs to some rank-maximal matching of $G_i$.

$s'_i$ and $C'$ are vertex disjoint.  When we apply the alternating path $s'_i$ to $M_i$ obtaining $N'_i$, it does not affect the vertices of $C'$. Hence, every edge of $C'$ belongs to some rank-maximal matching of $H_i$, because we can also apply $s'_i$ to a rank-maximal matching $M'_i=M_i \oplus C'$ of $G_i$ and obtain a rank-maximal matching $N''_i$ of $H_i$.

If  phase $i$ is augmenting, every path of $S_i$ path ends at a free vertex in $M_i$. After the application of $s_i$, $u$ is the only matched vertex in $Alive(i+1) \cup AV_i$ reachable from $C$. Hence, once again $u$ is one of the endpoints of $P'$. The rest of the proof is analogous to the previous case. This completes the proof.  \qed
\end{proof}

\begin{figure}
\centering
   \includegraphics[width=.5\textwidth]{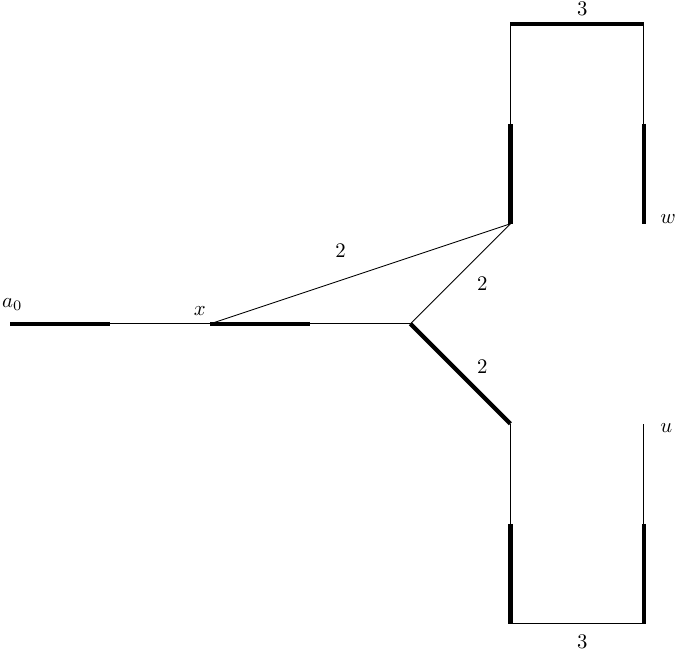}
	\caption{In the figure, $P'$ represents an alternating path from $u$ to $w$ via vertex $x$. Here $s_i$ is the alternating path from $a_0$ to $u$}
	\label{updatefig3}
\end{figure}

It remains to show how to efficiently compute $EG$-decompositions of graphs $H'_i$ given $\tilde{H}_i$. Note that we cannot simply apply Lemma \ref{hgDec} a multiple number of times for each of the graphs $\tilde{H}_i$, as such an approach would lead to an algorithm of complexity $\mathcal{O}(rm)$. Below we describe a general idea behind Algorithm \ref{updatesalgorithm} for computing $EG$-decompositions of $H'_i$, then in Theorem \ref{updatesruntimeproof} we show that it is possible to implement this algorithm to achieve an $\mathcal{O}(m + \min(c'n, n^2))$ runtime.

\begin{lemma} \label{uptype}

\begin{enumerate}
\item Free vertices in $N_i$ are the same in $H'_i$ and $\tilde{H}_i$.
\item If a vertex $v$ is reachable in $\tilde{H}_i$ from a free vertex in $N_i$ via an $N_i$-alternating path and ending at an alive vertex in $\tilde{H}_{i+1}$, then:
\begin{enumerate}
  \item $v$ is reachable in $\tilde{H}_j$ from a free vertex in $N_j$ via an $N_j$-alternating path ending at an alive vertex in $\tilde{H}_{j+1}$ for every $j < i$. 
 \item for every $j \leq i$ such that $v \in E(G'_j) \cup  O(G'_j)$ $v$ has the same type in $H'_j$ and $G'_j$, i.e., 
  $v \in E(G'_j) \Leftrightarrow  v \in E(H'_j)$ and  $v \in O(G'_j) \Leftrightarrow  v \in O(H'_j)$.
 
 \item for every $j < i$ such that $v \in U(G'_j)$, $v$ has the same type in $H'_j$ as in $H'_i$ if phases $i$ and $j$  are either both augmenting  or both non-augmenting 
and otherwise, $v \in O(H'_j) \Leftrightarrow v \in E(H'_i)$ and $v \in E(H'_j) \Leftrightarrow v \in O(H'_i)$.
\end{enumerate}

\end{enumerate}
\end{lemma}

\begin{proof}

We first prove $(1)$. From Theorem \ref{cor} point 3 we know that $N_i$ is a rank-maximal matching of $H_i$ and also a maximum matching of $H'_i$. Theorem \ref{cor} point 3 implies that $N_i$ is a maximum matching of $\tilde{H}_i$. Both $H'_i$ and $\tilde{H}_i$ have the same set of vertices, thus free vertices with respect to $N_i$ are the same in $H'_i$ and $\tilde{H}_i$ and $(1)$ holds.

$2(a)$. This part directly follows from Lemma \ref{hgDec}.

$2(b)$. Let $v \in E(G'_j) \cup O(G'_j)$. This assumption implies that $v \notin C$. From Lemma \ref{aux1} points \ref{aux1d} and \ref{aux2c} we know that $v$ has the same type in $G'_j$ and $\tilde{H}_j$. Additionally $v$ is reachable from a free vertex by an alternating path ending at an alive vertex in the graph $\tilde{H}_{j+1}$. By Lemma \ref{hgDec} $v \in E(\tilde{H}_j)$ if and only if $v \in E(H'_j)$. Similarly we have $v \in O(\tilde{H}_j)$ if and only if $v \in O(H'_j)$, hence $2(b)$ is proven.

$2(c)$. Let $v \in U(G'_j)$ for some $j < i$. There is a $N_i$-alternating path from a free vertex to an alive vertex in $\tilde{H}_{i+1}$ that contains $v$. 
Hence from part 2(a), there is a $N_j$-alternating path from a free vertex to $v$ in $\tilde{H}_j$, hence $v \in C$. 

Let $k$ be the phase during which $v$ is included to $C$. Since $C$ is augmented only in non-augmenting phases, $k$ must be non-augmenting.
After the addition of $v$ to $C$, there exists in $C$ an $M_k$-alternating path $T$ from $a_0$ to $v$. The length of $T$ is either odd
or even. We can notice that for each $j>k$, $T$ remains included in $C_j$, because we do not remove any edges from $C$, and it is an $M_j$-alternating path from $a_0$ to $v$, because $M_k \subseteq M_j$. Therefore, by Lemma \ref{aux1}, for each $j>k$, the type of $v$ in $\tilde{H}_j$ is the same as parity of $T$ iff $j$ is non-augmenting and otherwise (if $j$ is augmenting), its type in $\tilde{H}_j$ is opposite to the parity of $T$.
Also, the parity of $T$ is the same as $v$'s type in $\tilde{H}_k$.

By Lemma \ref{hgDec}, $v$ has the same type in $\tilde{H}_j$ and $H'_j$ for each $j < i$ and hence  $2(c)$ holds. \qed
\end{proof}

From Lemma \ref{uptype} we know that if $P'$ is a path in $\tilde{H}_i$ from a free vertex $u$ to a vertex $w \in Alive(i+1) \cup AV_i$, then we can determine types of all vertices of $P'$ in $H'_i$. Additionally from Lemma \ref{uptype} it is possible to determine types of such vertices in graphs $H'_j$ for each $j < i$. The above observations are a basis of Algorithm \ref{updatesalgorithm}. We start with $i=r$ and determine the set $Z_i$ of all vertices belonging to paths as described above (i.e. from a free vertex $u$ to a vertex $w \in Alive(i+1) \cup AV_i$).  Then we update the type of each vertex from $Z_i$ using Lemma \ref{uptype}, set $i \leftarrow i - 1$ and repeat the process for the new graph $\tilde{H}_i$. We continue iterating over $i$ until we reach $i = 0$. Note that if for some vertex $v$ we have $v \in Z_i$ and $v \notin Z_j$ for each $j > i$, then we have $v \in U(H'_j)$ for each $j > i$. Thus we can correctly determine types of all the vertices using this approach. 

Of course a naive implementation of the above idea does not achieve an $\mathcal{O}(m + \min(c'n, n^2))$ runtime.  Additional observations are needed. Let $v$ be any vertex. First we note that if $i$ is maximal such that $v$ belongs to $Z_i$ then types of $v$ in graphs $H'_1, H'_2, \ldots, H'_r$ can be correctly determined. There is no need to update the type of $v$ anymore even if it belongs to some $Z_j$ for $j < i$. Thus throughout the execution of the algorithm we maintain the set $Z$ of vertices for which we have already computed types and make sure to only update the types of vertices belonging to $Z_i \setminus Z$. In order to speed up the algorithm we need to show how to efficiently compute sets $Z_i \setminus Z$.

Let us first show how to find vertices belonging to $Z_i$. We first build an $N_i$-alternating forest of vertices reachable from the set $F_i$ of free vertices with respect to $N_i$. Then we determine the set $X$ of vertices belonging to $T_i$ and $Alive(i+1) \cup AV$. Next we consider a graph $(V(T_i), W_i)$ where $W_i$ is the set of edges with both endpoints in $T_i$. It is easy to see that all vertices reachable by alternating paths from $X$ in this graph form the set $Z_i$. Note that from Lemma \ref{uptype} it follows that $V(T_i) \subseteq V(T_j)$ for $j < i$, hence we do not have to build the alternating forest from scratch in each iteration. Instead for each $i$ we simply determine $T_i$ using the forest $T_{i+1}$. The set $Z_i \setminus Z$ can be determined similarly to the set $Z_i$. Instead of considering a graph $(V(T_i), W_i)$ we simply consider a graph $(V(T_i) \setminus Z, W_i)$ and claim that $Z_i \setminus Z$ is equal to the set of vertices reachable from $X$ in this graph. 

Computations of forests $T_i$ take $\mathcal{O}(m + min(nc', n^2))$ time in total. It is a straightforward consequence of the fact that $V(T_i) \subseteq V(T_{i-1})$. Similarly we can see that the time needed to compute all vertices reachable from $X$ in graphs $(V(T_i) \setminus Z, W_i)$ over the duration of the algorithm is also bounded by $\mathcal{O}(m + min(nc', n^2))$. It is a consequence of the fact that once a vertex $v$ is detected to be in $Z_i \setminus Z$ it is added to $Z$ and none of the edges incident to such a vertex is visited in any of the following iterations.

From the above discussion we obtain the correctness of the following lemma.

\begin{theorem}
\label{updatesruntimeproof}
Algorithm \ref{updatesalgorithm} computes $EG$-decompositions of graphs $H'_i$ in $\mathcal{O}(m + \min(c'n, n^2))$ time.
\end{theorem}




\begin{algorithm} [h]
\caption{for computing $EG$-decompositions of graphs $H'_i$ }
\label{updatesalgorithm}
\begin{algorithmic}[1]
\State $Z \leftarrow \emptyset$
\State $F \leftarrow \emptyset$
\State $N_{r+1} \leftarrow N_r$

\For {$i = r, r-1, \ldots, 1$}
   \State $N_i \leftarrow N_{i+1}  \setminus {\cal F}_{i+1}$
   \State $F_i \leftarrow$ free vertices with respect to $N_i$ \label{freeLine}
   \State $T_i \leftarrow$ an $N_i$-alternating forest in $\tilde{H}_i$ starting from vertices of $F_i$
   \State $W_i \leftarrow$ edges of $\tilde{H}_i$ with both endpoints in $T_i$
   \State $X \leftarrow$ vertices belonging to $T_i$ and $Alive(i+1) \cup AV$ 
   \State $Z_i \setminus Z \leftarrow$  all vertices reachable in $(V(T_i) \setminus Z, W_i)$ from $X$ via an $N_i$-alternating path \label{reachableLine}
   \ForAll {$v \in Z_i \setminus Z$}
	    \State for every $j>i$, $v \in U(H'_j)$ \label{updateX1}
			\State $v \in E(H'_i)$ (resp. $O(H'_i)$) if $v$ is reachable via an even (resp. odd) length $N_i$-alternating path \label{updateX2}
	    \If {$v \in E(G'_i) \cup O(G'_i)$} 
			   \State  for every $j \leq i$, $v$'s type in $H'_j$ is  the same  as in $G'_j$ \label{updateX3}
			\Else
			  \State for every $j <i$ such that $v \in U(G'_j)$, $v$'s type in $H'_j$ is  determined as in Lemma \ref{uptype} 2c \label{updateX4}
			\EndIf	
	 \EndFor
	\State $Z \leftarrow Z \cup Z_i$ \label{ZLine}
\EndFor
\ForAll {$v \in V \setminus Z$}
   \State $v \in U(H'_i)$ for every $i$ \label{updateX6}
\EndFor
\end{algorithmic}
\end{algorithm}

In the first part of this section, we detected and removed the edges that are present in $\tilde{H}$ but not in $H'$. There may also be some edges that are present in ${H}'$ but not in $\tilde{H}$. For example, in Figure \ref{updatefig2}$(a)$ we notice that , $a \in O(\tilde{H}_3)$ and $p \in U(\tilde{H}_3)$, hence the edge $(a,p)$ doesn't belong to $AE$. Therefore by Line $4$ of Algorithm $\ref{main}$, the edge is not added to $\tilde{H}_3$. But if we consider $H'_3$ (Figure \ref{updatefig2}(b)), $(a,p)$ is an $UU$ edge and is not deleted from the graph. Therefore to complete updating the reduced graphs of $H$ we need to include $(a,p)$ to $H'_3$. 

From Theorem \ref{cor} point $1$, we know that every edge belonging to some rank-maximal matching of $H$ is present in the graph $\tilde{H}$. Suppose the edge $(a,p)\notin \tilde{H}$. Then the edge can't be matched in any rank-maximal matching of $H$. But if $(a,p)$ is present in $H'$, then the edge must be of type $UU$. We know that every $EO$ edge in $H'$ lies in an alternating path from a free vertex to an alive vertex. By Theorem $1$ of \cite{GhoshalNN14}, every edge of the path belongs to some rank-maximal matching.

As a last step of updating the reduced graphs $H'_1, H'_2, \ldots$ we deal with edges belonging to $H'_i \setminus \tilde{H}_i$.
All such edges are of type $UU$ in $H'_i$. For each vertex $v$, we find the minimum rank $i$ such that $v \in O(\tilde{H}_i) \cap U(H'_i)$, if it exists. Next we find every edge incident to $v$ such that the other endpoint belongs to $O(\tilde{H}_i) \cup U(\tilde{H}_i)$. This edge is of type $OO \cup OU$ and is removed from $\tilde{H}$. Finally, if both endpoints of the edge belong to $U(H'_i)$, then we add it to $H'_i$. Identification of such edges take $\mathcal{O}(m)$ time in total as we check every edge at most once.  For each vertex $v$, we can find the minimum rank $i$ with $v \in O(\tilde{H}_i) \cap U(H'_i)$ in constant time using the ordered list that we maintain to store the GE-decomposition of the vertex in every reduced graphs of $H$.

\begin{figure}
\centering
   \includegraphics[width=\textwidth]{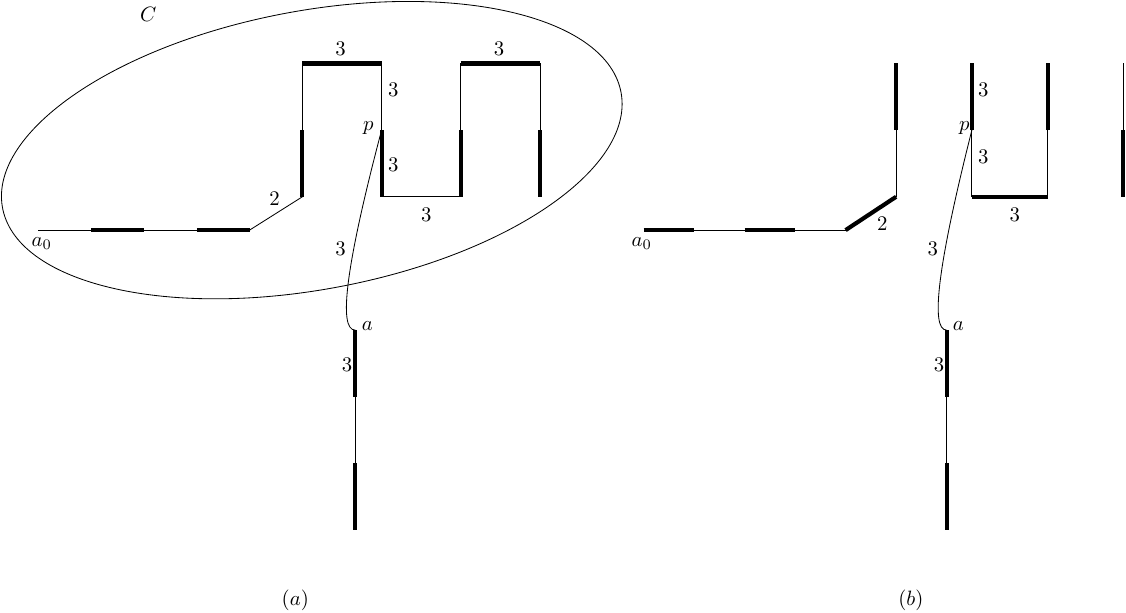}
	\caption{In Figure $(a)$, the edge $(a,p)$ is removed from the graph $\tilde{H}_3$, but in Figure $(b)$ the edge $(a,p)$ is a $UU$ edge in $H'_3$ and is not deleted from the graph}
	\label{updatefig2}
\end{figure}


\section{Remaining Update Operations} \label{remainingUpdates}

Let us remind that the algorithm for the dynamic version of the rank-maximal matching problem supports the following operations:

\begin{enumerate}
	\item Add a vertex $v$ along with incident edges to $G$
	\item Delete a vertex $v$ along with incident edges from $G$
	\item Add an edge $e$ to $G$
	\item Delete an edge $e$ from $G$
\end{enumerate}

We have already shown how to implement operation $(1)$. In this section we prove that $(2)-(4)$ can be essentially reduced to $(1)$. The following lemma is crucial for the reduction. 

\begin{lemma} \label{del-lemma}
Let $G$ be a an instance of the rank-maximal matching problem. Let $a_1$ and $p_1$ be two vertices of $G$ such that $(a_1, p_1)$ is matched in every rank-maximal matching of $G$. If $M$ is a rank-maximal matching of $G$ then $M \setminus \{(a_1, p_1)\}$ is a rank-maximal matching of the graph $G \setminus \{a_1, p_1\}$.
\end{lemma}
\begin{proof}
Let $M$ be a rank-maximal matching of $G$. Since $(a_1, p_1)$ is a matched edge in $G$, $M \setminus \{(a_1, p_1)\}$ is a matching of $G \setminus \{a_1, p_1\}$. Suppose $M \setminus \{(a_1, p_1)\}$ is not a rank-maximal matching of $G \setminus \{a_1, p_1\}$. We assume that $M'$ is a rank-maximal matching of $G \setminus \{a_1, p_1\}$. Hence $M'$ has a strictly better signature than $M \setminus \{(a_1, p_1)\}$ in $G \setminus \{a_1, p_1\}$. If we consider the matching $M' \cup \{(a_1, p_1)\}$, it is a matching of $G$ and it has a strictly better signature than $M$ in $G$. But $M$ is a rank-maximal matching of $G$. Therefore we arrive at a contradiction. Thus $M \setminus \{(a_1, p_1)\}$ is a rank-maximal matching of the graph $G \setminus \{a_1, p_1\}$.  \qed
\end{proof}

\subsection{Deletion of a Vertex}
In order to implement operation $(2)$ we modify the graph so that we can apply operation $(1)$ to delete a vertex from $G$. Let us assume that we want to delete an applicant $a_1$ from $G$. We introduce a dummy post $p_d$ and a rank `zero'. A post has rank `zero' in the preference list of an applicant if he prefers that post more than his rank $1$ post. We add $p_d$ as a rank `zero' post in the preference list of $a_1$. We define the new graph as $H$. Note that the graph $H_0$ contains only one edge $(a_1, p_d)$. We  calculate a rank-maximal matching of $H$ from a rank-maximal matching of $G$ with the aid of operation $(1)$. Let $M$ be a rank-maximal matching of $H$. It can easily be verified that $(a_1, p_d)$ is matched in every rank-maximal matching of $H$. By Lemma \ref{del-lemma}, the matching $M \setminus \{(a_1, p_d)\}$ is a rank-maximal matching of $H \setminus \{a_1, p_d\}$ and obviously also a rank-maximal matching of $G \setminus \{a_1\}$. In order to delete a post $p_1$ from $G$, we proceed analogously but this time we add a dummy applicant $a_d$ to the graph instead of a dummy post.

\subsection{Addition and Deletion of an Edge}
Both operations $(3)$ and $(4)$ can be implemented in a very similar way using operations $(1)$ and $(2)$. We first show how to implement operation $(3)$, i.e. how to add an edge to the graph. Let us assume that we wish to add an edge $(a, p)$ to the graph $G$. In order to do so we first use the operation $(2)$ to delete the vertex $a$ along with its incident edges from $G$. Next we simply use the operation $(1)$ to add the vertex $a$ again, but this time the incidence list of this vertex is larger by one, i.e. it contains all the edges incident to the old "copy" of $a$ along with the new edge $(a, p)$.

Operation $(4)$ can be implemented in an analogous way.
\section{Dynamic Popular Matching} \label{popular}

In this section, we give a simple reduction which allows us to use our dynamic rank-maximal matching algorithm to solve the dynamic popular matching problem. First we formally define popular matching. Let $G = (\mathcal{A} \cup \mathcal{P}, \mathcal{E})$ be a bipartite graph and $a \in \mathcal{A}$ be an applicant. For two matchings $M$ and $M'$ of $G$ we say that $a$ prefers $M$ to $M'$ if either $a$ is matched in $M$ and unmatched in $M'$, or $rank(a, M(a)) < rank(a, M'(a))$. 

\begin{definition}
A matching $M$ is said to be more popular than $M'$ if the number of applicants preferring $M'$ is no more than the number of applicants preferring $M$. A matching $M$ is said to be popular in $G$ if that matching is more popular than any other matching of $G$.
\end{definition}

As mentioned in \cite{AbrahamIKM07}, an unique last resort post $l(a)$ is added to each applicant $a$ as their least preferred post. For an applicant $a$, $f(a)$ denotes the set of rank $1$ posts adjacent to $a$. These posts are called $f$ posts. And $s(a)$ denotes the set of most preferred posts of $a$ belonging  to $E(G_1)$, where $G_1$ is the subgraph of $G$ containing rank $1$ edges. Abraham et al. \cite{AbrahamIKM07} proved the following theorem.

\begin{theorem} \label{popmain}
A matching $M$ is popular in $G$ iff
\begin{enumerate}
\item $M \cap \mathcal{E}_1$ is a maximum matching of $G_1 = (\mathcal{A} \cup \mathcal{P}, \mathcal{E}_1)$, 
\item for each agent $a$, $M(a) \in \{f(a) \cup s(a)\}$.
\end{enumerate}
\end{theorem}

Given an instance of popular matching $G$, we are going to introduce $G_{RMM}$, such that we can find a popular matching of $G$ by computing a rank-maximal matching in $G_{RMM}$. We define $G_{RMM} = (\mathcal{A} \cup \mathcal{P} , \mathcal{D}_1 \cup \mathcal{D}_2)$ where sets of vertices in $G$ and $G_{RMM}$ are identical and each edge has rank $i \in \{1,2\}$. Here as $\mathcal{D}_1$ we denote rank $1$ edges in $G$. Edges of $\mathcal{D}_1$ are rank $1$ in $G_{RMM}$. For each $a \in \mathcal{A}$, $\mathcal{D}_2$ contains the most preferred edges incident to $a$ that are also incident to a post belonging to $E(G_{RMM,1})$. We assign the ranks to the edges of $\mathcal{D}_2$ in the following way. If an edge is rank $1$ in $G$, then it is rank $1$ in $G_{RMM}$, otherwise we set the rank of this edge to $2$.

In order to compute a popular matching of $G$, we first calculate a rank-maximal matching of $G_{RMM}$. Next we check if the matching of $G_{RMM}$ is applicant complete or not. A matching is called an applicant complete matching if it matches every applicant present in the graph. If the matching is an applicant complete matching, then we claim that it is also a popular matching of $G$, otherwise no popular matching exists in $G$.  

\begin{lemma}
Let $G$ be a bipartite graph and we calculate $G_{RMM} = (\mathcal{A} \cup \mathcal{P} , \mathcal{D}_1 \cup \mathcal{D}_2)$ from $G$ by the reduction described above. If we can compute an applicant complete rank-maximal matching of $G_{RMM}$, then that is also a popular matching of $G$. Otherwise $G$ does not contain a popular matching.
\end{lemma}

\begin{proof}
First, we can notice that the graphs $G_1$ and $G_{RMM,1}$ are identical. After the first iteration of rank-maximal matching on $G_{RMM}$, we compute a maximum matching of $G_{RMM,1}$, which satisfies the first condition of Theorem \ref{popmain}. To prove the second condition, we observe the edges present in $G_{RMM}$. The edges belonging to $\mathcal{D}_1$ are between an applicant and his $f$ posts.  Since we have shown that $G_1$ and $G_{RMM,1}$ are identical, the Gallai Edmonds decomposition of $G_1$ and $G_{RMM,1}$ are the same. Thus, $\mathcal{D}_2$ contains the edges between each applicant and his $s$ post. Therefore, for each $a$, $G_{RMM}$ contains the edges incident to $\{f(a), s(a)\}$. Therefore any applicant complete rank-maximal matching satisfies the second condition of Theorem \ref{popmain}. Thus a rank-maximal matching of $G_{RMM}$ is indeed a popular matching of $G$. \qed
\end{proof}

It is easy to observe that $G_{RMM}$ is the same graph as the reduced graph that we get during the combinatorial algorithm \cite{AbrahamIKM07} for popular matching. Only the ranks of the edges may not be the same in these two graphs. In the next algorithm, we give a pseudocode for dynamic popular matching with the help of the algorithm for dynamic rank-maximal matching.

\begin{algorithm} 
\caption{Dynamic Popular Matching}
\label{mainpop}
\begin{algorithmic}[1]
\State Construct the graph $G' = (\mathcal{A} \cup \mathcal{P}, \mathcal{E}
)$, where $\mathcal{E} = \{(a, p) | p \in f(a) \cup s(a), a \in \mathcal{A}\}$
\State Compute the graph $G_{RMM}$ by reassigning the ranks of the edges of $G'$
\State Let $M$ be a popular matching of $G$ and a rank-maximal matching of $G_{RMM}$
\State $H = G \cup \{a\}$
\State $H_{RMM} = G_{RMM} \cup \{a\}$, where $a$ is  isolated 
\State Add the edges corresponding to $f$ posts of $a$ to $H_{RMM}$
\State Perform the first iteration of Algorithm \ref{main} and update the Gallai-Edmonds decomposition of the graph $H_{RMM,1}$
\State Update  $s$ posts of each applicant and add the edges corresponding to newly found $s$ posts to $\mathcal{D}_2$ 
\State Assign appropriate ranks to the newly added edges of $\mathcal{D}_2$
\State Perform the second iteration of Algorithm \ref{main} on $H_{RMM}$ and update the Gallai Edmonds decomposition
\If {the rank-maximal matching of $H_{RMM}$ is an applicant complete rank-maximal matching}
	\State The rank-maximal matching of $H_{RMM}$ is a popular matching of $H$
\Else
	\State $H$ does not have a popular matching\EndIf
\end{algorithmic}
\end{algorithm}

As we can see above it may happen that at some point during the execution of the dynamic algorithm no popular matching exists. It is important to emphasise that our algorithm maintains a rank-maximal matching of $G_{RMM}$ regardless of whether a popular matching exists or not. Then the existence of a popular matching can be easily checked based on the rank-maximal matching of $G_{RMM}$. In particular it may happen that at some point as a result of an update operation a previously unsolvable instance $G$ becomes solvable. Note that in this case we do not have to compute a popular matching from scratch as we already have a precomputed rank-maximal matching of $G_{RMM}$ which is not applicant complete. We simply update such a matching and obtain an applicant complete matching of $G_{RMM}$ which is popular in $G$. 
\newpage
\bibliography{rmm_lipics}
\bibliographystyle{plain}

\end{document}